\documentclass[acmtog]{acmart}

\AtBeginDocument{%
  }

\setcopyright{cc}
\setcctype{by}
\copyrightyear{2026}
\acmYear{2026}
\acmDOI{10.1145/3829340.3842222}
\acmISBN{979-8-4007-2842-6/2026/12}
\acmArticle{000}
\acmConference[SA Conference Papers '26]{SIGGRAPH Asia 2026 Conference Papers}{December 01--04, 2026}{Kuala Lumpur, Malaysia}
\acmBooktitle{SIGGRAPH Asia 2026 Conference Papers (SA Conference Papers '26), December 01--04, 2026, Kuala Lumpur, Malaysia}

\acmSubmissionID{71}

\usepackage{booktabs} 
\usepackage{verbatim}
\usepackage{multirow}
\usepackage{graphicx}
\usepackage{amsmath}
\usepackage{enumerate}
\usepackage{color}
\usepackage{alltt}
\usepackage{listings}
\usepackage{subcaption}
\usepackage{acmart-taps}
\usepackage[ruled]{algorithm2e}

\definecolor{highlightvone}{RGB}{0,0,0}
\definecolor{highlightvtwo}{RGB}{0,0,0}
\definecolor{highlightvthree}{RGB}{0,0,0}

\long\def\highlightv#1#2{{
  \ifcase#1\color{highlightvone}
  \or\color{highlightvone}
  \or\color{highlightvtwo}
  \or\color{highlightvthree}
  \else\color{highlightvone}
  \fi
  #2}}

\begin{document}
\title{Semi-Implicit Pairwise Descent for Nonlocal Continuum Mechanics}

\author{Xukun Luo}
\authornote{Equal contribution.}
\orcid{0009-0003-6194-3328}
\email{luoxukun2022@iscas.ac.cn}
\affiliation{%
  \institution{Institute of Software, Chinese Academy of Sciences and UCAS}
  \city{Beijing}
  \country{China}}

\author{Xiao Cheng}
\authornotemark[1]
\orcid{0009-0006-4582-4742}
\email{chengxiao25@ios.ac.cn}
\affiliation{%
  \institution{Institute of Software, Chinese Academy of Sciences and UCAS}
  \city{Beijing}
  \country{China}}

\author{Yuzhong Guo}
\orcid{0009-0000-4578-9701}
\email{guoyuzhong@iscas.ac.cn}
\affiliation{%
  \institution{Institute of Software, Chinese Academy of Sciences}
  \city{Beijing}
  \country{China}}

\author{Ying Qiao}
\orcid{0000-0002-0002-2328}
\email{qiaoying@iscas.ac.cn}
\affiliation{%
  \institution{Institute of Software, Chinese Academy of Sciences}
  \city{Beijing}
  \country{China}}

\author{Wencheng Wang}
\orcid{0000-0001-5094-4606}
\email{whn@ios.ac.cn}
\affiliation{%
  \institution{Institute of Software, Chinese Academy of Sciences}
  \city{Beijing}
  \country{China}}

\author{Xiaowei He}
\authornote{Corresponding author.}
\orcid{0000-0002-8870-2482}
\email{xiaowei@iscas.ac.cn}
\affiliation{%
  \institution{Institute of Software, Chinese Academy of Sciences}
  \city{Beijing}
  \country{China}}

\makeatletter
\let\SIPD@mkauthorsaddresses\@mkauthorsaddresses
\def\@mkauthorsaddresses{UCAS: University of Chinese Academy of Sciences.\par\par
  \SIPD@mkauthorsaddresses}
\makeatother

\renewcommand{\shortauthors}{Luo et al.}

\begin{abstract}
We propose Semi-Implicit Pairwise Descent (SIPD), a unified nonlocal pairwise framework for simulating large-scale hyperelastic materials involving complex contact and friction. 
By reformulating the Finite Element Method (FEM) equations of motion into a pairwise force representation from a nonlocal perspective, our approach avoids costly Hessian computations, leading to a reduction in per-iteration computational overhead.
\highlightv1{Furthermore, we propose an analytical projection strategy for projecting our Hessian-free coefficient matrices to positive semi-definiteness.}
And we treat contact and friction as a unified anisotropic elastic energy, allowing for a seamless integration into the elastic solver framework.
We mathematically prove that our method is unconditionally stable and numerically convergent.
Experimental results demonstrate that SIPD achieves real-time performance for million-scale simulations even under intricate contact and friction conditions.
\end{abstract}


\begin{teaserfigure}
\centering
\includegraphics[width=0.28\linewidth]{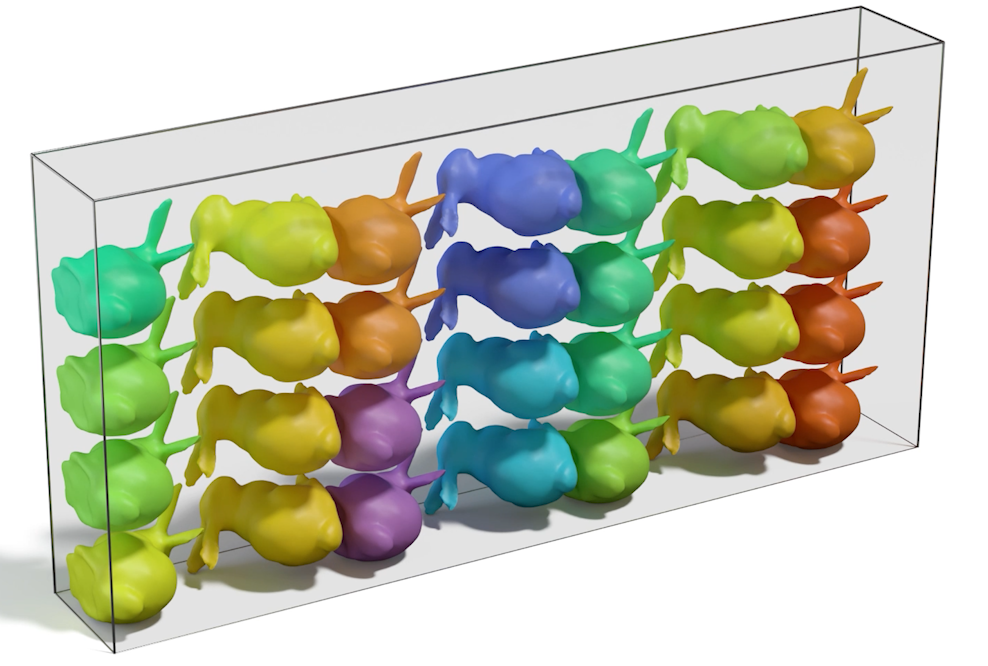}
\includegraphics[width=0.28\linewidth]{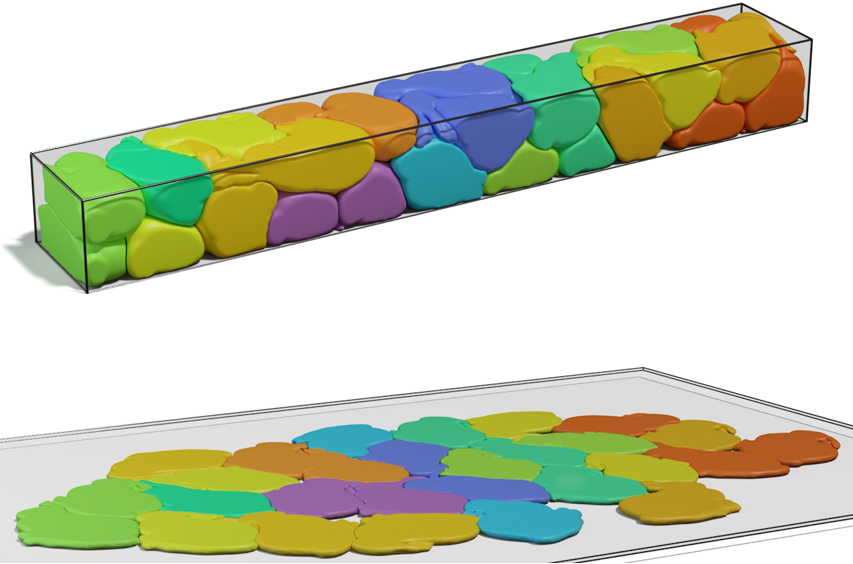}
\includegraphics[width=0.42\linewidth]{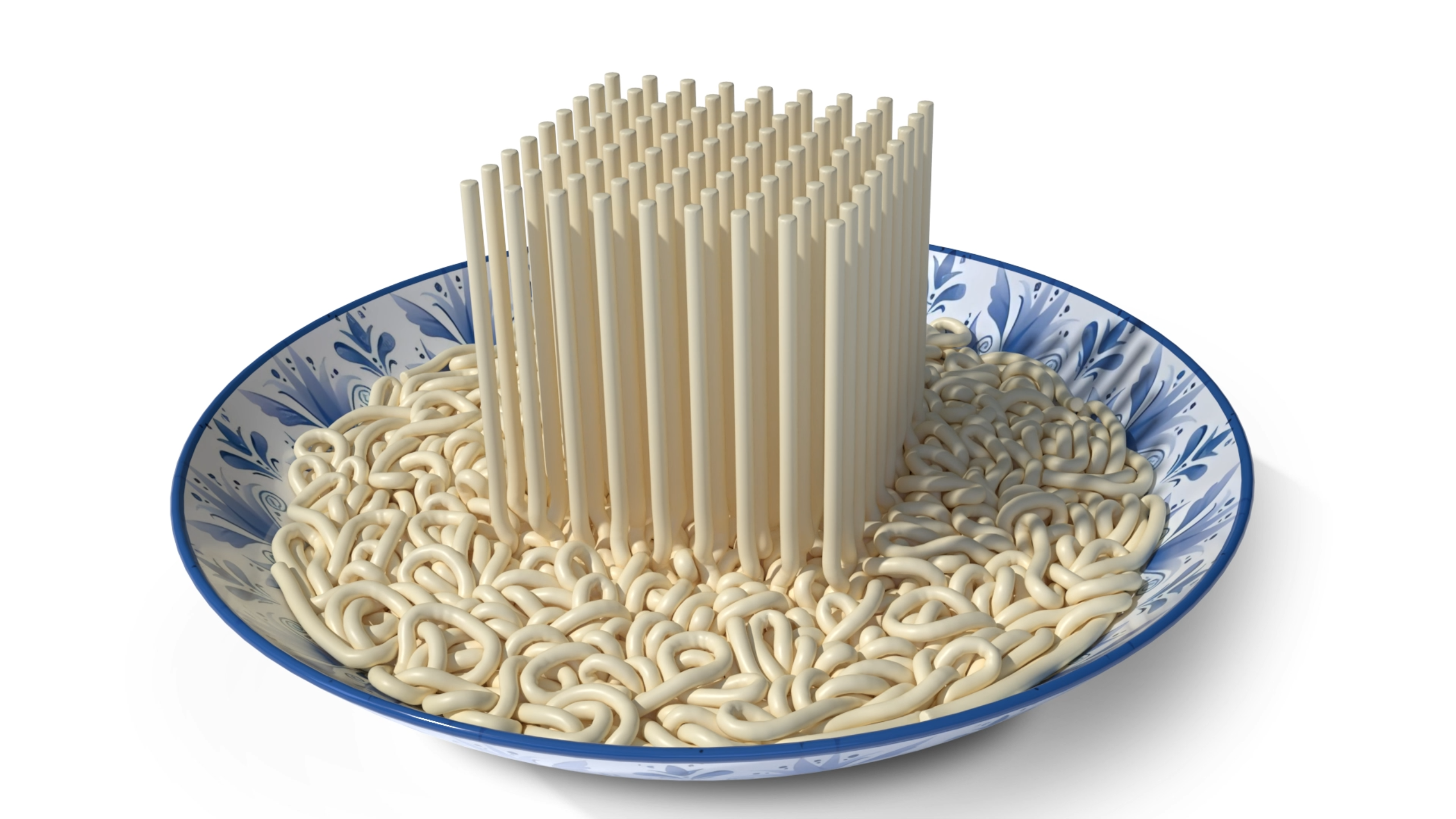}
\caption{\textbf{Left}: 28 bunnies (1.03M tetrahedra, 225K vertices) are compressed and deformed by the walls of a transparent box upon impact. Our solver maintains \textbf{real-time} performance with an average cost of 3ms per step ($h=3ms$). \textbf{Right}: Noodles (1.8M tetrahedra, 697K vertices) falling into a bowl with 2.62M active collision pairs, computed at 14.3ms per step ($h=3ms$).}
\label{fig:fall}
\end{teaserfigure}

\maketitle

\section{Introduction}
The fast, robust, and accurate simulation of hyperelastic materials involving discontinuous contact forces has long been a fundamental goal in computer graphics applications. 
Since explicit time integration schemes impose strict limitations on the time step size, semi-implicit~\cite{Bridson:2005:SCF,Schroeder:2011:AEF} and fully implicit~\cite{Liu:2013:FSM,Macklin:2016:XPBD,Smith:2018:SNH} time integration schemes are increasingly favored for real-time applications. 
Moreover, with the significant increase in GPU computing power over the past decade, the development of implicit solvers capable of fully leveraging GPU advantages for deformable body dynamics becomes a new trend.
However, fully addressing the challenges of simulating hyperelastic materials with both nonlinear behavior and discontinuous contact forces on GPUs requires not only a straightforward extension of traditional numerical models but also a deeper understanding and reformulation of the underlying problem.

Within classical continuum mechanics, the behavior of hyperelastic materials is typically modeled using partial differential equations.
\highlightv2{When the continuum is discretized into tetrahedra or other control volumes using finite element analysis, the governing equations reduce to a finite nonlinear system in terms of the nodal positions~\cite{Irving:2004:IFE,teran2005robust,Sifakis:2012:FEM}.}
To solve the nonlinear system, which may include various constraints arising from contact, sequential algorithms typically employ the Newton-Raphson solver. 
This solver begins by linearizing the elastic forces and then calculates the energy Hessian. 
However, ensuring that the stiffness matrix remains positive definite is challenging, particularly when arbitrary forms of discontinuous contact forces are included~\cite{shi1992discontinuous,KANG20256977}.
Although a straightforward implementation of the Newton-Raphson solver on GPUs can significantly enhance performance due to increased computational power, the algorithm is not well suited to GPU architectures because of the high dynamic memory overhead associated with storing the Hessian matrix.
To address this issue, Wang and Yang~\shortcite{wang2016descent} propose using a diagonalized Hessian matrix to calculate the descent direction, ensuring positive definiteness by simply checking the values of the diagonal elements. 
Chen et al.~\shortcite{chen2024vertex} instead propose Vertex Block Descent (VBD) and derive a $3 \times 3$ Hessian matrix for each vertex to update each vertex's position by solving a small local linear system. However, they still need to calculate the fourth-order tensors 
\highlightv2{by taking the derivative of the first Piola–Kirchhoff stress tensor $\mathbf{P}$ with respect to the deformation gradient $\mathbf{F}$}, i.e., $\partial \mathbf{P}/\partial \mathbf{F}$.
\highlightv2{Moreover, VBD does not guarantee a descent direction because each vertex block may not be positive definite.}

We present the very GPU-friendly semi-implicit pairwise descent to solve the nonlinear optimization problem.
Unlike the Newton-Raphson solver that requires calculating the second-order derivative of the energy function, our key innovation relies on reformulating the variational problem into a nonlocal peridynamics framework depending only on the first-order derivative of the energy density function, in which both internal and boundary forces are represented in a unified pairwise form. 
\highlightv2{Thus, we efficiently construct the $3 \times 3$ coefficient matrix for each vertex using only the first Piola--Kirchhoff stress $\mathbf{P}$ and the deformation gradient $\mathbf{F}$ obtained from the constitutive model.}
During each iteration, a novel semi-implicit splitting based on~\cite{lu2023projective} and an SVD-free analytical projection strategy are employed to ensure that the positive component of the coefficient matrix used for implicit integration remains positive definite. 
Consequently, the $3 \times 3$ coefficient matrix assembled at each vertex is guaranteed to be positive definite.
Our method can uniformly accommodate scenarios involving discontinuous contact and friction, achieving real-time simulation for million-scale scenes.

\section{Related Work}
Continuum mechanics centers on describing the macroscopic behavior of matter through partial differential equations (PDEs). 
Since Terzopoulos et al.~\shortcite{terzopoulos1987elastically} pioneered this theory to computer graphics, the Finite Element Method (FEM) has become the standard for deformation modeling.
Sin et al.~\shortcite{sin2013vega} subsequently extended this framework to hyperelastic models within an implicit Euler integration scheme.
Kugelstadt et al.~\shortcite{kugelstadt2018fast} employed operator splitting to separately \highlightv3{solve} the shear and bulk energies for efficient computation.
Schneider et al.~\shortcite{schneider2018decoupling} proposed an adaptive scheme for adjusting the order of finite element basis functions, achieving high-fidelity results even on coarse meshes.
While FEM excels at modeling continuous deformations, it inherently struggles with physical discontinuities like fracture and contact, where spatial derivatives become ill-defined.

To bridge this gap, the nonlocal continuum mechanics, exemplified by Peridynamics (PD)~\cite{silling2007peridynamic}, replaces differential operators with integral operators using a pairwise function $\psi(\mathbf{x}, \mathbf{x}')$ within a specific horizon. 
This formulation effectively captures \highlightv3{discontinuous} behavior, including crack initiation, strain localization, and fracture.
Leveraging nonlocal theory, Levine et al.~\shortcite{levine2014peridynamic} reformulated mass-spring systems to achieve efficient brittle fracture simulation. 
Similarly, Chen et al.~\shortcite{chen2018peridynamics} developed a plastic model within the PD framework for ductile fracture, and Han et al.~\shortcite{han2021lagrangian} introduced a nonlocal diffusion model for temperature-dependent fracture simulation. 
Regarding solution strategies, He et al.~\shortcite{he2017projective} proposed a corotational peridynamics formulation based on Projective Dynamics, while Lu et al.~\shortcite{lu2023projective} extended this to hyperelastic materials and developed a semi-implicit acceleration scheme.

Nevertheless, PD models often lack the precision of FEM in representing continuous behaviors, leading researchers~\cite{dorduncu2024review} to explore coupling strategies that utilize both models to handle continuous and discontinuous regions respectively. 
For instance, Yu et al.\shortcite{yu2025implicit} developed a coupling strategy based on force equilibrium, employing FEM for elastic continua and PD for crack regions. 
This coupling paradigm has also informed the integration of FEM with the Material Point Method (MPM)~\cite{li2024dynamic}. 
However, these hybrids often face challenges in numerical stability and implementation complexity.

\highlightv2{Contact is inherently characterized by discontinuous behavior, manifested as displacement discontinuities across interfaces. 
Thus early research typically treat collisions as impulsive forces~\cite{bridson2002robust}. 
Modern approaches treat contact as a continuous optimization problem, via penalty methods~\cite{gast2015optimization,chen2024vertex}, dynamic penalty functions~\cite{harmon2009asynchronous,narain2012adaptive,tang2018cloth,giles2025augmented}, or Incremental Potential Contact (IPC)~\cite{li2020incremental,li2020codimensional}.
These nonlinear optimization problems are often solved using full Newton methods~\cite{li2020incremental,huang2024gipc}. 
A related class of quasi-Newton elastodynamic solvers combines L-BFGS low-rank curvature updates with lagged second-order information. 
Liu et al.~\shortcite{liu2017quasi} reused an FEM Hessian approximation across iterations and employed L-BFGS updates to incorporate current curvature information. 
Subsequent extensions employed more structured initializers, including a domain-decomposed model for mesh-based elastodynamics and a multigrid model for MPM, while retaining the L-BFGS iteration~\cite{li2019decomposed,wang2020hierarchical}.
More recently, Lu and Hu~\shortcite{lu2025reliable} introduced Reliable Iterative Dynamics (RID), which combines a force-splitting scheme with a dual-descent formulation to reduce linearization errors and produce reliable intermediate updates.
Other parallel solution strategies include position-based nonlinear Gauss--Seidel~\cite{chen2023position}, preconditioned nonlinear conjugate gradient~\cite{shen2024preconditioned}, Jacobi preconditioning based on the diagonal Hessian~\cite{wang2016descent}, and VBD, which performs local vertex updates through graph coloring~\cite{chen2024vertex}. 
These methods provide effective strategies for improving convergence, robustness, and parallel efficiency, but still rely on second-order Hessian information, either directly or through lagged or structured approximations. Motivated by reducing this dependency, we seek to solve the nonlinear system using only first-order constitutive information.}

\section{Foundations}
Since the purpose of this work is to develop a unified model capable of handling both continuous and discontinuous material behaviors, the following two sections first review fundamental concepts in both classical and nonlocal continuum mechanics.
To avoid notational ambiguity, the symbol $\mathbf{x}$ is used to denote coordinates in the reference configuration, whereas $\mathbf{y}$ denotes coordinates in the deformed configuration. 
Vertices with lumped mass in a tetrahedral mesh are uniformly referred to as particles.
\subsection{Finite Element Method}
\label{sec:fem}
The core idea of the finite element method (FEM) is to approximate the solution of a continuous problem by breaking the domain into a finite number of elements and solving a discretized variational problem over these elements~\cite{wriggers2008nonlinear}.
For a given continuum medium, the \highlightv3{equation of motion} is expressed as
\begin{equation}
\rho \ddot{\mathbf{y}} = \nabla _{\mathbf{x}} \cdot \mathbf{P} + \mathbf{b},
\label{eq:momentum}
\end{equation}
where $\nabla _{\mathbf{x}} \cdot \mathbf{P}$ denotes the divergence of the first Piola–Kirchhoff stress tensor $\mathbf{P}$ with respect to the reference coordinates ${\mathbf{x}}$, and $\mathbf{b}$ denotes the external force.
When the continuum medium is discretized into a tetrahedral mesh, the deformed coordinates $\mathbf{y}$ \highlightv3{are} approximated using shape functions $N_i(\mathbf{x})$: $\mathbf{y} =\sum\limits_{i} {N _i \left( \mathbf{x} \right){\mathbf{y}_i}},$
where $i$ denotes the vertex index and $\mathbf{y}_i$ is the deformed position of vertex $i$.
The shape function $N_i(\mathbf{x})$ satisfies the Kronecker-$\delta$ property~\cite{koschier2017robust} that $N_i(\mathbf{x}_j) = \delta _{ij}$.
The deformation gradient on each tetrahedral element $e$ then becomes 
\begin{equation}
\mathbf{F}_e(\mathbf{x}) = \sum\limits_{i \in \mathcal{V}_e} \mathbf{y}_i \otimes \nabla N_{i}(\mathbf{x}),
\label{eq:FinFEM}
\end{equation}
where $\mathcal{V}_e$ \highlightv3{denotes} the set of vertices of tetrahedral element $e$.
\highlightv2{Substituting this approximation into the weak form of Equation~\ref{eq:momentum} and integrating over the tetrahedral mesh yields the following discretized equation of motion for particle $i$~\cite{Sifakis:2012:FEM}:}
\begin{equation}
m_i \ddot{\mathbf{y}}_i = - \sum_{e \in \mathcal{T}_i} V_e \mathbf{P}_e \nabla N_{e} + \mathbf{b}_i
\label{eq:gov_local}
\end{equation}
\highlightv2{where $\mathcal{T}_i$ denotes the set of tetrahedral elements incident to particle $i$, $m_i$ and $\mathbf{b}_i$ denote the lumped mass and external force, and $V_e$ is the volume of element $e$.}

\subsection{Peridynamics}
As a typical nonlocal continuum theory, peridynamics replaces spatial derivatives with integral operators, and material points interact directly with other points within a finite distance. As a result, the governing equations remain valid even when the displacement field is discontinuous.
The \highlightv3{equation of motion} in a state-based model is expressed as~\cite{silling2007peridynamic}:
\begin{equation}
\rho \ddot{\mathbf{y}} = \int_{\mathcal{H}_\mathbf{x}} \{ \underline{\mathbf{T}}_\mathbf{x}\langle{\mathbf{x}' - \mathbf{x}}\rangle-\underline{\mathbf{T}}_{\mathbf{x}'}\langle{\mathbf{x}-\mathbf{x}'}\rangle \} dV_{\mathbf{x}'} + \mathbf{b},
\end{equation}
where $\mathcal{H}_\mathbf{x}$ is a spherical neighborhood centered at $\mathbf{x}$, and $\underline{\mathbf{T}}_\mathbf{x}$ and $\underline{\mathbf{T}}_{\mathbf{x}'}$ are the force vector states at points $\boldsymbol{\mathbf{x}}$ and $\mathbf{x'}$, respectively.
A state of order $m$ in peridynamics is defined to be a function that maps a bond $\xi = \mathbf{x}' - \mathbf{x}$ to a tensor of order $m$, typically denoted as $\underline{\mathbf{A}}\langle \xi \rangle$.
\highlightv2{Intuitively, a bond is simply the reference-space vector connecting two material points, and a force state maps each such bond to the interaction force transmitted along that connection. The neighborhood integral therefore accumulates pairwise interactions, much as FEM assembly accumulates the element forces shared by neighboring vertices.}
A significant difference of the momentum equation in a nonlocal continuum is that the force vector is defined as a function of two material points.
The same idea is also extended to define other physical quantities.
By introducing the deformation vector state $\underline{\mathbf{Y}}\left\langle \xi  \right\rangle = \mathbf{y}' - \mathbf{y}$, the deformation gradient in the nonlocal form is expressed as
\begin{equation}
\mathbf{F} = {\int_{{\mathcal{H}_\mathbf{x}}} {\underline{\omega}\left\langle \boldsymbol{\xi}  \right\rangle \underline{\mathbf{Y}}\left\langle \boldsymbol{\xi} \right\rangle  \otimes {\mathbf{K}^{ - 1}} \boldsymbol{\xi} d{V_\mathbf{x}}} },
\label{eq:FinPD}
\end{equation}
where $\underline{\omega}\left\langle \xi  \right\rangle$ is the weighting function that maps each bond $\xi$ to a scalar value, $\mathbf{K}$ denotes the shape tensor formulated as follows
\begin{equation}
\mathbf{K} = {\int_{{\mathcal{H}_\mathbf{x}}} {\underline{\omega}\left\langle \boldsymbol{\xi}  \right\rangle \boldsymbol{\xi} \otimes \boldsymbol{\xi} d{V_\mathbf{x}}} }
\label{eq:K}.
\end{equation}
By comparing Equation~\ref{eq:FinPD} with Equation~\ref{eq:FinFEM} through dimensional analysis, it can be observed that both equations share the same structure in computing the deformation gradient. 
The key difference is that Equation~\ref{eq:FinFEM} employs local variables ($\mathbf{y}_i$ and $\nabla N_i(\mathbf{x})$), whereas Equation~\ref{eq:FinPD} relies on nonlocal variables ($\underline{\mathbf{Y}}\langle \boldsymbol{\xi} \rangle$ and $\mathbf{K}^{-1}\boldsymbol{\xi}$) to evaluate the deformation gradient. 
Substituting the nonlocal expression of the deformation gradient into the momentum equation and integrating over $\mathcal{H}_x$ yields the discretized momentum equation in a nonlocal form expressed as follows~\cite{lu2023projective}:
\begin{equation}
m_i \ddot{\mathbf{y}}_i = V_i \sum\limits _j \omega_{ij} V_j \bigg(\mathbf{P}_i \mathbf{K}_i^{-1} + \mathbf{P}_j \mathbf{K}_j^{-1}\bigg) \boldsymbol{\xi} _{ij} + \mathbf{b}_i.
\end{equation}
where the subscripts $i$ and $j$ are used to denote particle $i$ and $j$, the bond $\xi_{ij} = \mathbf{x}_j - \mathbf{x}_i$.
Other variables such as $V$ and $\mathbf{P}$ have the same physical meanings as those in Equation~\ref{eq:gov_local}, except they are now defined on particles.

\section{A Nonlocal Pairwise Model for Continuous Deformation}
\highlightv1{In this section, we reformulate the FEM using the nonlocal concept of peridynamics introduced in Section 3.2. The resulting formulation is theoretically equivalent to the traditional FEM.}

\subsection{Global Optimization in a Nonlocal Formulation}
Consider an object that undergoes a deformation that maps every material point from its reference space $\mathbf{x}$ to the deformed space $\mathbf{y}=\vec{\phi}(\mathbf{x})$, where $\vec{\phi}$ is an arbitrary continuous function.
Without loss of generality, we consider the tetrahedron shown in Figure~\ref{fig:control} as a Lagrangian finite control volume, where the vertex positions in the reference and deformed configurations are denoted as $\mathbf{x}_{i=1, 2, 3, 4}$ and $\mathbf{y}_{i=1, 2, 3, 4}$.
\highlightv2{Assuming that the deformation map is linear within tetrahedron $e$, the lumped force on particle $i$ is derived from the finite element formulation in Section~\ref{sec:fem}:}
\begin{equation}
\mathbf{f}_{ie} = - V_e \mathbf{P}_e \nabla N_{ie},
\label{eq:fi}
\end{equation}
\highlightv2{where $N_{ie}(\mathbf{x})$ is the linear shape function associated with vertex $i$ and satisfies $N_{ie}(\mathbf{x}_j) = \delta _{ij}$.}
In classical continuum mechanics, $\mathbf{f}_{ie}$ can be interpreted as the total force acting on the boundary triangle opposite to vertex $i$, as shown in Figure~\ref{fig:control}(a).
To reformulate $\mathbf{f}_{ie}$ in a nonlocal form, we introduce the following theorem.

\begin{theorem}
Given a tetrahedron with four vertex positions denoted by $\mathbf{x}_{i=1, 2, 3, 4}$ and linear shape functions $N_{i}(\mathbf{x})$ satisfying $N_{i}(\mathbf{x}_j) = \delta _{ij}$, we have
\begin{equation}
\nabla {N_i} =  - \sum\limits_{j \in \mathcal{V}_e} {\mathbf{K}^{ - 1}_{e}} {{\boldsymbol{\xi} _{ij}}}, {\kern 5pt} \mathbf{K}_{e} = \sum\limits_{j \in \mathcal{V}_e} {{\boldsymbol{\xi} _{ij}} \otimes {\boldsymbol{\xi} _{ij}}}.
\label{eq:split}
\end{equation}
\end{theorem}
\begin{proof}
    See Supplementary Material, Section~A.
\end{proof}
Note $\mathbf{K}_{e}$ is a special case of the shape tensor defined in Equation~\ref{eq:K}, obtained by applying a constant weight function $\underline{\omega}\left\langle \boldsymbol{\xi}  \right\rangle = 1$.
After substituting Equation~\ref{eq:split} into Equation~\ref{eq:fi}, the particle force $\mathbf{f}_{ie}$ is decomposed into three pairwise forces, as illustrated in Figure~\ref{fig:control}(b). 
The physical interpretation of $\mathbf{f}_{ie}$ can thus be understood as the actions exerted by the other particles.
\highlightv2{After summing over all tetrahedra incident to particle $i$, the governing equation is reformulated as}
\begin{equation}
m_i \ddot{\mathbf{y}}_i = \sum\limits _j \bigg( \sum_{e \in \mathcal{T}_i \bigcap \mathcal{T}_j} V_e \mathbf{P}_e \mathbf{K}_e^{-1} \bigg) \boldsymbol{\xi} _{ij} + \boldsymbol{b}_i,
\label{eq:gov}
\end{equation}
where $V_e$, $\mathbf{P}_e$ and $\mathbf{K}_e$ denote physical quantities associated with the tetrahedron that shares the edge connecting particle $i$ and $j$.
Accordingly, the total force acting on particle $i$ by its neighboring particle $j$ is formulated in the following nonlocal form:
\begin{equation}
\mathbf{f}_{ij} = \bigg( \sum\limits_{e \in \mathcal{T}_i \bigcap \mathcal{T}_j} V_e \mathbf{P}_e \mathbf{K}_e^{-1} \bigg) (\mathbf{x}_j - \mathbf{x}_i).
\label{eq:force1}
\end{equation}

Since $\mathbf{f}_{ij}$ may contain nonlinear terms, Equation~\ref{eq:gov} is preferably solved using the optimization-based implicit Euler time integration method~\cite{gast2015optimization}.
Thus, the system can be converted into the following optimization problem:
\begin{equation}
{\mathcal L}({\mathbf{y}}) = \frac{1}{{2{h^2}}}\sum\limits_i {{m_i}||{{\mathbf{y}}_i} - {\mathbf{y}}_i^*|{|^2}}  + \sum\limits_e {V_e \Psi_e (\mathbf{y}) } ,
\label{eq:obj}
\end{equation}
where ${\mathbf{y}}_i^* = {\mathbf{y}}_i + h{\mathbf{v}}_i$, and $\Psi_e$ represents the strain energy density function.

\subsection{A Pairwise Force Model for Hyperelastic Continua}
In hyperelasticity, the strain energy density function $\Psi$ is typically formulated as a scalar function of deformation gradient $\mathbf{F}$, i.e., $\Psi=\Psi(\mathbf{F})$.
\highlightv1{For isotropic materials, $\Psi$ can be further simplified as a function of the three principal invariants of Cauchy-Green deformation tensor $\mathbf{C} = \mathbf{F}^T\mathbf{F}$, i.e., $\Psi=\Psi(I_1,I_2,J)$. Following~\cite{smith2019analytic}, the first Piola-Kirchhoff stress $\mathbf{P}$ is then derived as:
\begin{equation}
    \mathbf{P} = \frac{\partial \Psi }{\partial \mathbf{F}} = 2\frac{\partial\Psi}{\partial I_1}\mathbf{F} + 2\frac{\partial\Psi}{\partial I_2}(I_1\mathbf{F}-\mathbf{FC})+J\frac{\partial\Psi}{\partial J}\mathbf{F}^{-T}.
    \label{eq:iso}
\end{equation}
The above expression for $\mathbf{P}$ not only simplifies the calculation, but also makes it easier to guarantee the convergence for iterative methods in solving the nonlinear optimization problem in Equation~\ref{eq:obj}. More details will be discussed in Section~\ref{sec:APD}.
}

\highlightv2{Within element $e$, linear interpolation gives $\mathbf{y}_j-\mathbf{y}_i=\mathbf{F}_e\boldsymbol{\xi}_{ij}$ and hence $\mathbf{F}_e^{-1}(\mathbf{y}_j-\mathbf{y}_i)=\boldsymbol{\xi}_{ij}$. Substituting this identity into Equation~\ref{eq:force1} yields the following pairwise force between particles $i$ and $j$:}
\begin{equation}
{\mathbf{f}_{ij}} = \sum\limits_{e \in \mathcal{T}_i \bigcap \mathcal{T}_j} V_e {\mathbf{P}_e} \mathbf{K}{_e^{ - 1}} {\mathbf{F}_e^{-1}} { {(\mathbf{y}_j - \mathbf{y}_i)}}.
\label{eq:f_pos}
\end{equation}

\textbf{Remove the shape tensor $\mathbf{K}$}: To further simplify the formulation, we propose a novel solution to eliminate the inverse of the shape tensor $\mathbf{K}_e^{-1}$ by introducing the following identical relation:
\begin{equation}
    \sum_{j \in \mathcal{V}_e} \mathbf{K}_e^{-1} \boldsymbol{\xi}_{ij} = \sum_{j \in \mathcal{V}_e} w_{ij,e} \boldsymbol{\xi}_{ij}
    \label{eq:identity}
\end{equation}
where $w_{ij,e}$ represents the geometric weight associated with the edge $ij$ of element $e$ (different from the weight function $\underline{\omega}\left\langle \xi  \right\rangle$).
When the edge vectors $\{\boldsymbol{\xi}_{i1}, \boldsymbol{\xi}_{i2}, \boldsymbol{\xi}_{i3}\}$ are linearly independent (non-degenerate tetrahedron), they form a complete basis for $\mathbb{R}^3$. 
Consequently, the vector $\mathbf{K}_e^{-1} \sum_j \boldsymbol{\xi}_{ij}$ can be uniquely expressed as the linear combination $\sum_j w_{ij,e} \boldsymbol{\xi}_{ij}$. 
A practical solution in calculating the weights for each tetrahedron is provided in section 2 of the supplemental document.

Specifically, the weight vector $\mathbf{w}_{ie}$ for each tetrahedron vertex can be calculated as $\mathbf{w}_{ie} = \mathbf{G}_{ie}^{-1}\mathbf{1}$, where $\mathbf{G}_{ie}$ is the Gram matrix of vertex $i$ and $\mathbf{1} = [1, 1, 1]^T$.
Inserting Equation~\ref{eq:identity} into Equation~\ref{eq:f_pos} yields a more concise expression for the pairwise force between particles $i$ and $j$:
\begin{equation}
{\mathbf{f}_{ij}} = \underbrace{\sum\limits _{e \in \mathcal{T}_i \bigcap \mathcal{T}_j} w_{ij,e} \left(  V_e \mathbf{P}_e\mathbf{F}_e^{-1} \right)}_{{\mathbf{L}}_{ij} \in \mathbb{R}^{3\times3}} { {(\mathbf{y}_j - \mathbf{y}_i)}}.
\label{eq:f_ij}
\end{equation}
\highlightv2{Hereafter, ${\mathbf{L}}_{ij}$ is referred to as the \emph{pairwise tension tensor}, as it has the dimensions of energy per unit area.}
\highlightv2{As shown in Supplementary Material, Section~B, $w_{ij,e}=w_{ji,e}$; therefore, the pairwise force model in Equation~\ref{eq:f_ij} satisfies Newton's third law, i.e., $\mathbf{f}_{ij}=-\mathbf{f}_{ji}$.}

\begin{figure}[h]
	\centering
        \includegraphics[width=\linewidth]{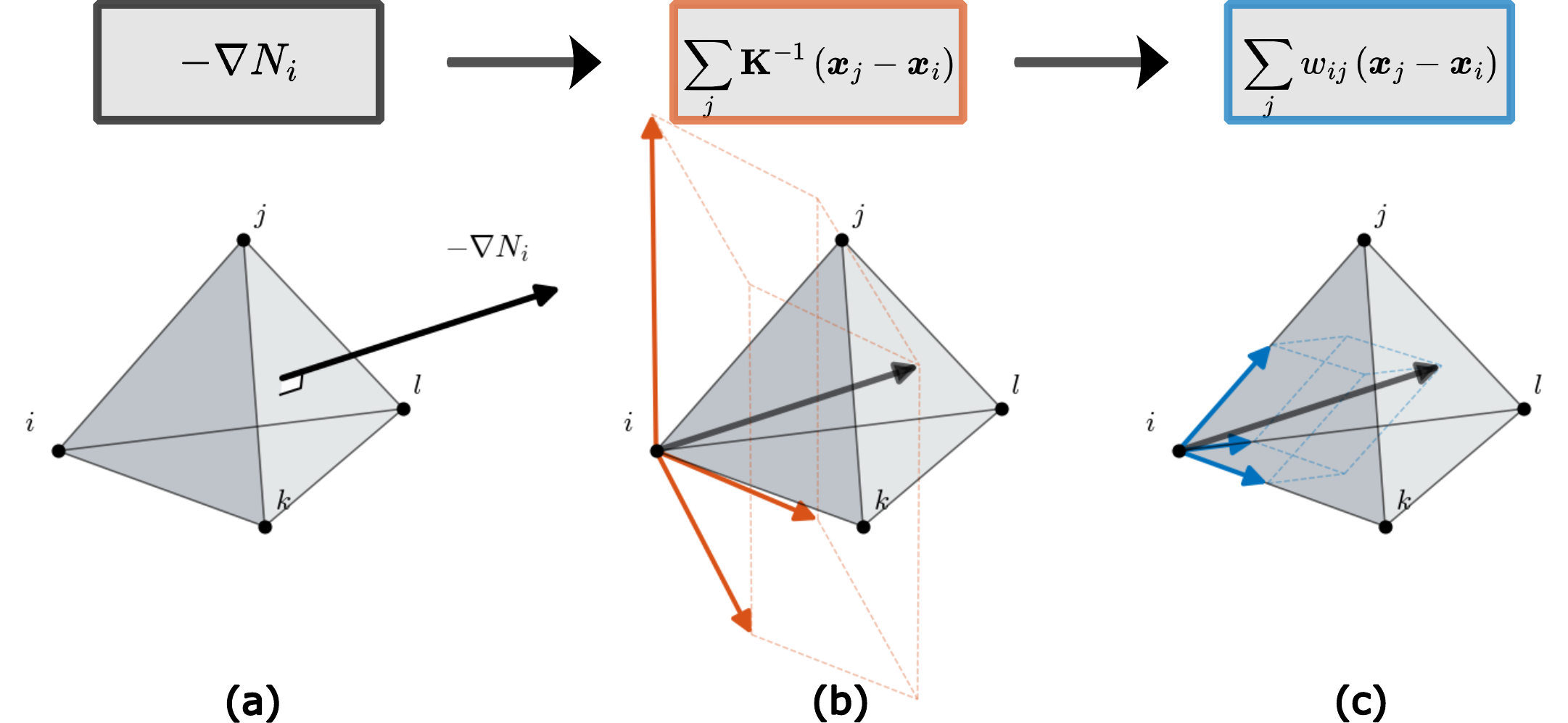}

	\caption{Perspectives on force modeling across different simulation frameworks. (a) the FEM perspective; (b) the peridynamic perspective; and (c) our proposed perspective.}
	\label{fig:control}
\end{figure}

The nonlocal pairwise force model described above offers several key advantages over classical continuum mechanics-based models, particularly when applied to hyperelastic materials undergoing discontinuous deformation or contact.
First, the structure of ${\mathbf{L}}_{ij}$ significantly simplifies the numerical modeling of arbitrary hyperelastic materials, as it involves only the first-order derivative of the energy density function with respect to the principal stretches of the deformation gradient.
\highlightv2{Second, discontinuous deformations involving contact and friction can be represented consistently by virtual tetrahedra, as detailed in Section~\ref{sec:contact}.}
Finally, the implementation of an implicit solver is highly compatible with modern GPU architectures.
This is because ${\mathbf{L}}_{ij}$ can be efficiently computed using only $\mathbf{P}_e,\mathbf{F}_e \in \mathbb{R}^{3\times3} $ and ${w}_{ij,e}\in \mathbb{R}$, without requiring the costly evaluation of $\partial \mathbf{P}/\partial \mathbf{F} \in \mathbb{R} ^{9\times9}$.
Moreover, symmetry and positive definiteness can be easily guaranteed via a semi-implicit decomposition and an analytical projection strategy, as demonstrated in Section~\ref{sec:SIPD}.

\section{Semi-Implicit Pairwise Descent}
\label{sec:SIPD}

Given both the internal and boundary forces in the nonlocal pairwise form, the position-based governing equation derived from Equation~\ref{eq:obj} and Equation~\ref{eq:f_ij} can be uniquely expressed as follows:
\begin{equation}
    \frac{m_i}{h^2} (\mathbf{y}_i - \mathbf{y}_i^*) - \sum_j \underbrace{\sum_{e \in \mathcal{T}_i \bigcap \mathcal{T}_j} w_{ij,e} \left(  V_e \mathbf{P}_e\mathbf{F}_e^{-1} \right)}_{\mathbf{L}_{ij}} { {(\mathbf{y}_j - \mathbf{y}_i)}} = 0.
    \label{eq:m_ij}
\end{equation}

\highlightv2{To solve the nonlinear system above, a Newton-type method linearizes the nonlinear term by introducing the $9\times9$ Hessian $\partial\mathbf{P}/\partial\mathbf{F}$, which is expensive to compute and difficult to project to a positive-definite matrix.}
\highlightv2{We instead solve the nonlinear system using the \textit{successive substitution method}~\cite{langtangen2019introduction}, a Jacobi-style iterative method that linearizes the nonlinear term by keeping it fixed at each iteration ($\mathbf{L} = \mathbf{L}^k$), without computing the Hessian.}

\highlightv2{In this Jacobi-style iteration, all $\mathbf{y}_j$ are fixed at $\mathbf{y}_j^k$ while computing the update $\Delta\mathbf{y}_i^{k+1}=\mathbf{y}_i^{k+1}-\mathbf{y}_i^k$ from the following linear system:}
\begin{equation}
    (\frac{m_i}{h^2} \mathbf{I} + \sum_j \mathbf{L}_{ij}^k) \Delta \mathbf{y}_i^{k+1} = -\overbrace{\bigg(\frac{m_i}{h^2}(\mathbf{y}_i^k-\mathbf{y}^{*}_i)-\sum_j\mathbf{L}^k_{ij}(\mathbf{y}_j^k-\mathbf{y}_i^k)\bigg)}^{\nabla \mathcal{L}_i(\mathbf{y}^k)}.
    \label{eq:Axb}
\end{equation}
Thus, the update $\Delta \mathbf{y}_i^{k+1}$ can be solved by $-(\frac{m_i}{h^2} \mathbf{I}+ \sum_j \mathbf{L}_{ij}^k)^{-1} \nabla \mathcal{L}_i(\mathbf{y}^k)$. 
Consequently, the position of each particle can be updated independently using only information from neighboring particles.
By bypassing the Hessian evaluation inherent to second-order techniques (such as VBD method), our first-order method incurs significantly lower computational and register overheads, making it ideal for parallel implementation.
\highlightv1{Furthermore, an analytical projection strategy is proposed to ensure convergence.}

\subsection{Positive Definiteness and Convergence}
\label{sec:APD}
Note when the $3\times 3$ matrix ${\mathbf{L}}_{ij}$ is positive semi-definite, the coefficient matrix ${\left(\frac{m_i}{h^2} \mathbf{I} + \sum_j \mathbf{L}_{ij}^k \right)^{ - 1}}$ is \highlightv1{accordingly} positive definite, we have 
\begin{equation}
\begin{array}{l}
\begin{aligned}
\nabla \mathcal{L}^k_i \cdot \left( {{{\mathbf{y}}^{k + 1}} - {{\mathbf{y}}^k}} \right) =  - \nabla \mathcal{L}^k_i \cdot {\bigg(\frac{m_i}{h^2} \mathbf{I} + \sum_j \mathbf{L}_{ij}^k \bigg)^{ - 1}}\nabla \mathcal{L}^k_i \leq 0.
 \end{aligned}
\end{array}
\end{equation}
The direction defined by $\Delta \mathbf{y}^{k+1} = \mathbf{y}^{k+1} - \mathbf{y}^k$ is clearly a descent direction whenever $\lVert \nabla \mathcal{L}^k \rVert \neq 0$. 

However, the positive definiteness of ${\mathbf{L}}_{ij}$ is not guaranteed under all conditions, particularly during compression. 
To address this issue, we present an extension of the semi-implicit successive substitution method (SISSM), originally introduced by Lu et al.~\shortcite{lu2023projective} for original peridynamics, to our reformulated FEM model. 
Its core idea is to decompose the coefficient matrix into positive and negative components (i.e., $\mathbf{L}_{ij} = \mathbf{L}_{ij}^+ + \mathbf{L}_{ij}^-$), treating them implicitly and explicitly, respectively. 
Specifically, the explicit treatment of the negative components means they are excluded from the coefficient matrix corresponding to the implicit variables on the left-hand side of Equation~\ref{eq:Axb}.

Considering the definiteness of the pairwise tension tensor for each finite element is jointly defined by the values of $w$ and ${\mathbf{P}}{\mathbf{F}}^{-1}$, our solution to decompose the pairwise tension tensor $\mathbf{L}$ is as follows:
\begin{equation}
\begin{array}{l}
\begin{aligned}
{\mathbf{L}_{ij,e}^ + } &= V_e  \left( {{w_{ij,e}^ + } {{({\mathbf{P}}{{{\mathbf{F}}}^{ - 1}})_e^+} }  + {w_{ij,e}^ - } {{( {\mathbf{P}}{{{\mathbf{F}}}^{ - 1}})}_e^ - }} \right) \\
{\mathbf{L}_{ij,e}^ - } &= V_e  \left( {{w_{ij,e}^ - } {{( {\mathbf{P}}{{ {\mathbf{F}}}^{ - 1}})}_e^ + }  + {w_{ij,e}^ + } {{( {\mathbf{P}}{{ {\mathbf{F}}}^{ - 1}})_e}^ - }} \right)
\end{aligned}
\end{array}
\label{eq:decompose}
\end{equation}
where the decomposition for $w$ is done as follows
\begin{equation}
{w^ + } = \left\{ {\begin{array}{*{20}{c}}
\begin{aligned}
{w,{\kern 5pt}}&{w > 0}\\
{0,{\kern 5pt}}&{w \le 0}
\end{aligned}
\end{array}} \right. , {\kern 10pt}{w^ - } = \left\{ {\begin{array}{*{20}{c}}
\begin{aligned}
{0,{\kern 5pt}}&{w > 0}\\
{w,{\kern 5pt}}&{w \le 0}
\end{aligned}
\end{array}} \right. ,
\end{equation}
and the $(\mathbf{P}\mathbf{F}^{-1})$ should be decomposed into a positive definite part and a negative definite part. This decomposition can be achieved by traditional eigenvalue projection methods based on singular value decomposition (SVD), such as the eigenvalue clamping projection strategy~\cite{teran2005robust} and the absolute eigenvalue projection strategy~\cite{chen2024stabler}.
Although applying these projection methods to our $3\times 3$ first Piola-Kirchhoff stress matrix $\mathbf{P}$ is much faster than applying them to a $9\times 9$ or $12 \times 12$ Hessian matrix, we still want to avoid the costly SVD.

\highlightv1{Thus, we propose an analytical projection strategy without SVD for isotropic hyperelastic materials, which is obtained by multiplying Equation~\ref{eq:iso} with $\mathbf{F}^{-1}$}:
\begin{equation}
\mathbf{P}\mathbf{F}^{-1} = 2\frac{\partial\Psi}{\partial I_1}\mathbf{I} + 2\frac{\partial\Psi}{\partial I_2}(I_1\mathbf{I}-\mathbf{F}^T\mathbf{F})+J\frac{\partial\Psi}{\partial J}(\mathbf{FF}^T)^{-1}.
\end{equation}
\highlightv1{The above equation can be decoupled into two components: the positive definite matrices including $\mathbf{I}$, $\mathbf{F}^T\mathbf{F}$ and $(\mathbf{FF}^T)^{-1}$ as well as the scalar values including $\frac{\partial\Psi}{\partial I_1}$, $\frac{\partial\Psi}{\partial I_2 }I_1$ and $J\frac{\partial\Psi}{\partial J}$, all of which can be efficiently computed without SVD.
The above observation allows us to easily partition $(\mathbf{P}\mathbf{F}^{-1})$ into a positive definite part $(\mathbf{P}\mathbf{F}^{-1})^{+}$ and a negative definite part $(\mathbf{P}\mathbf{F}^{-1})^{-}$ based on their values.}

\highlightv1{
For example, we present an decomposition for the stable neo-hookean energy (Equation (13) in ~\cite{Smith:2018:SNH}) in case of $J \ge 0 $ as follows:
\begin{equation}
\begin{aligned}
{\mathbf{P}\mathbf{F}^{-1}} &= \mu \mathbf{I}  - \mu J(\mathbf{F}\mathbf{F}^{T})^{-1} + \lambda (J - 1)J(\mathbf{F}\mathbf{F}^{T})^{-1}\\
(\mathbf{P}\mathbf{F}^{-1})^{+} &= \mu\mathbf{I} + \lambda J^2 (\mathbf{F}\mathbf{F}^{T})^{-1}\\
(\mathbf{P}\mathbf{F}^{-1})^{-} &= -(\mu + \lambda)J (\mathbf{F}\mathbf{F}^{T})^{-1}.
\end{aligned}
\label{eq:Ldecompose}
\end{equation}
}
\highlightv1{
For the deteriorated case where $J<0$, we can set $(\mathbf{PF}^{-1})^{+}=\mathbf{PF}^{-1}$ and $(\mathbf{PF}^{-1})^{-}=0$ to preserve positive definiteness.
}
\highlightv2{Therefore, the stable Neo-Hookean model~\cite{Smith:2018:SNH} remains well-defined for inverted and nearly singular configurations. 
In our practical implementation, when $\mathbf F$ becomes singular, we evaluate the gradient term $w_{ij}\mathbf P\mathbf F^{-1}(\mathbf y_j-\mathbf y_i)$ using its equivalent form $w_{ij}\mathbf P(\mathbf x_j-\mathbf x_i)$, thereby avoiding an explicit evaluation of $\mathbf F^{-1}$. When calculating the term $\mathbf L$, we use the following formula
\begin{equation}
(\mathbf F\mathbf F^T)^{-1}_{\epsilon}
:=\frac{\operatorname{adj}(\mathbf F\mathbf F^T)}{\max(J^2,\epsilon^2)},
\end{equation}
to prevent the pairwise coefficient matrix from becoming unbounded as $J\rightarrow0$.}

After inserting Equation~\ref{eq:decompose} into Equation~\ref{eq:Axb}, the update can be reformulated into $\Delta \mathbf{y}_i^{k+1} = -(\frac{m_i}{h^2} \mathbf{I} + \sum_j \mathbf{L}_{ij}^{k+})^{-1} \nabla \mathcal{L}_i(\mathbf{y}^k)$, with $\mathbf L_{ij}^{k-}$ treated explicitly on the right-hand side.
\highlightv2{As shown in Supplementary Material, Section~D, the analytical projection strategy makes every local update matrix positive definite and hence produces a descent direction. Combined with line search, it ensures convergence. Figures~\ref{fig:stretch} and~\ref{fig:APD} further demonstrate its stability and line-search behavior in practice.}
\highlightv2{The detailed SIPD procedure is provided in Algorithm~\ref{alg:sipd}.}
\section{Anisotropic Modeling for Contact and Friction}
\label{sec:contact}
To uniformly model contact and friction within our nonlocal pairwise model, we introduce virtual boundary elements (VBE) for proximity pairs whose distances
are below a given activation threshold $\hat{d}$, motivated by Wang et al.~\shortcite{Wang:2023:FGT} and M\"{u}ller et al.~\shortcite{Muller:2015:AMR}.

Without loss of generality, we consider a vertex-triangle proximity pair in Figure~\ref{fig:vbe}(a), where $(\mathbf{y}_0, \mathbf{y}_1, \mathbf{y}_2, \mathbf{y}_3)$ denote the vertex coordinates in the deformed configuration.
\highlightv2{The corresponding vertex coordinates in the reference configuration, where the contact and friction energies are zero, are defined as}
\begin{equation}
\begin{array}{l}
\begin{aligned}
{{\mathbf{x}}_0} &= {{\mathbf{y}}_0} - ({d_n} - \hat d){\mathbf{n}} - {d_t }\mathbf{t},
{\kern 10pt}
{{\mathbf{x}}_{1,2,3}} = {{\mathbf{y}}_{1,2,3}},
\end{aligned}
\end{array}
\end{equation}
where $d_n$ denotes the signed distance measuring the normal deformation along the normal vector $\mathbf{n}$, while $d_t$ denotes the unsigned distance measuring the magnitude of the shearing deformation along the tangential vector $\mathbf{t}$.
\highlightv2{To control the normal and tangential deformations independently, we construct two virtual tetrahedra, as shown in Figure~\ref{fig:vbe}(b) and Figure~\ref{fig:vbe}(c).}
\highlightv2{The deformation gradients and their inverses for the normal and tangential components are then formulated separately as}
\begin{equation}
\begin{array}{l}
\begin{aligned}
\mathbf{F}_n &= \mathbf{I} + (\frac{d_n}{\hat{d}} - 1)\mathbf{n}\mathbf{n}^T, {\kern 5pt} \mathbf{F}_n^{-1} = \mathbf{I} + (\frac{\hat{d}}{d_n} - 1)\mathbf{n}\mathbf{n}^T\\
\mathbf{F}_t &= \mathbf{I} + \frac{d_t}{\hat{d}} \mathbf{t}\mathbf{n}^T, {\kern 29pt} \mathbf{F}_t^{-1} = \mathbf{I} - \frac{d_t}{\hat{d}} \mathbf{t}\mathbf{n}^T
\end{aligned}
\end{array}
\end{equation}

\begin{figure}
    {\includegraphics[width=0.32\linewidth]{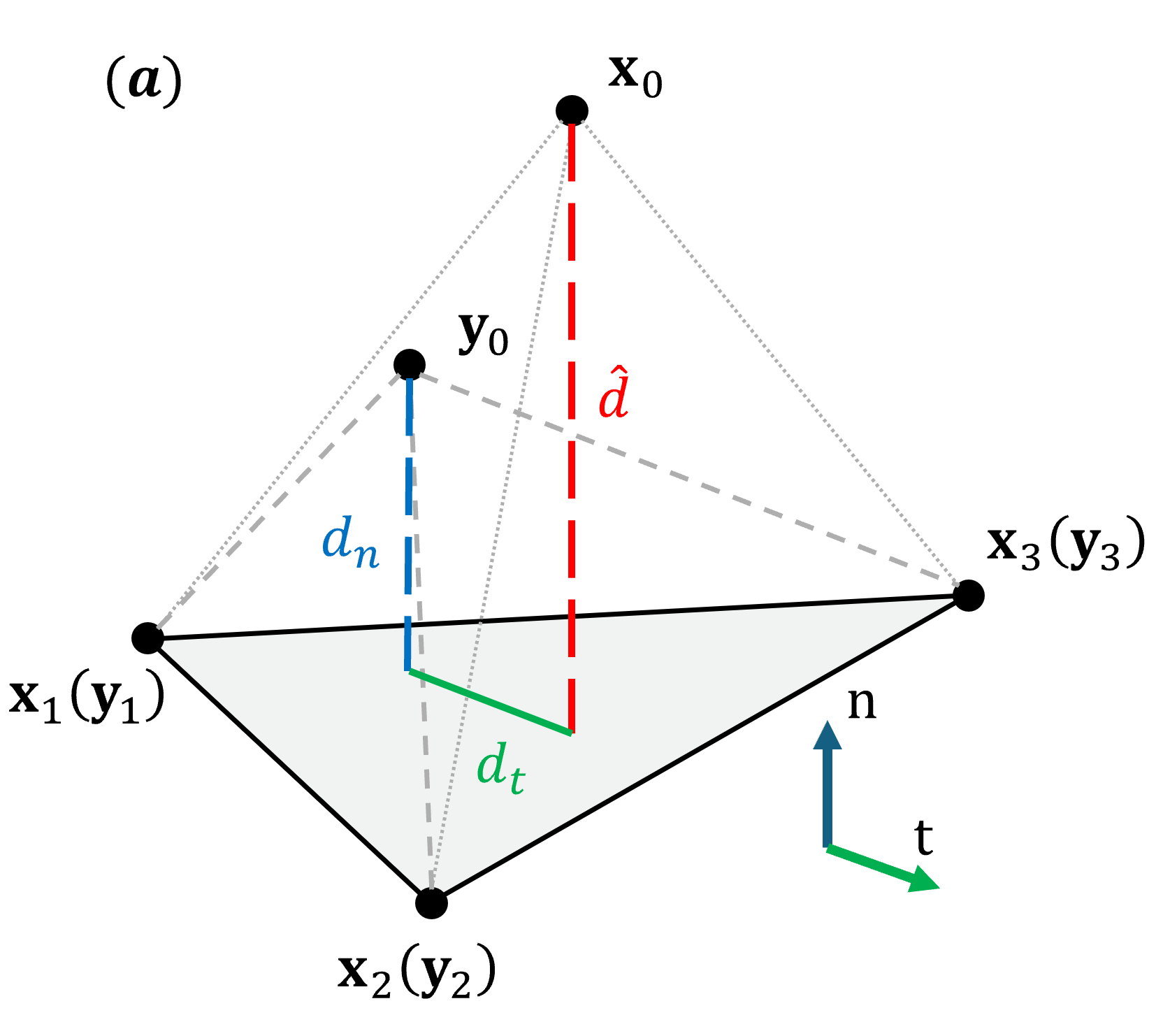}}
    {\includegraphics[width=0.32\linewidth]{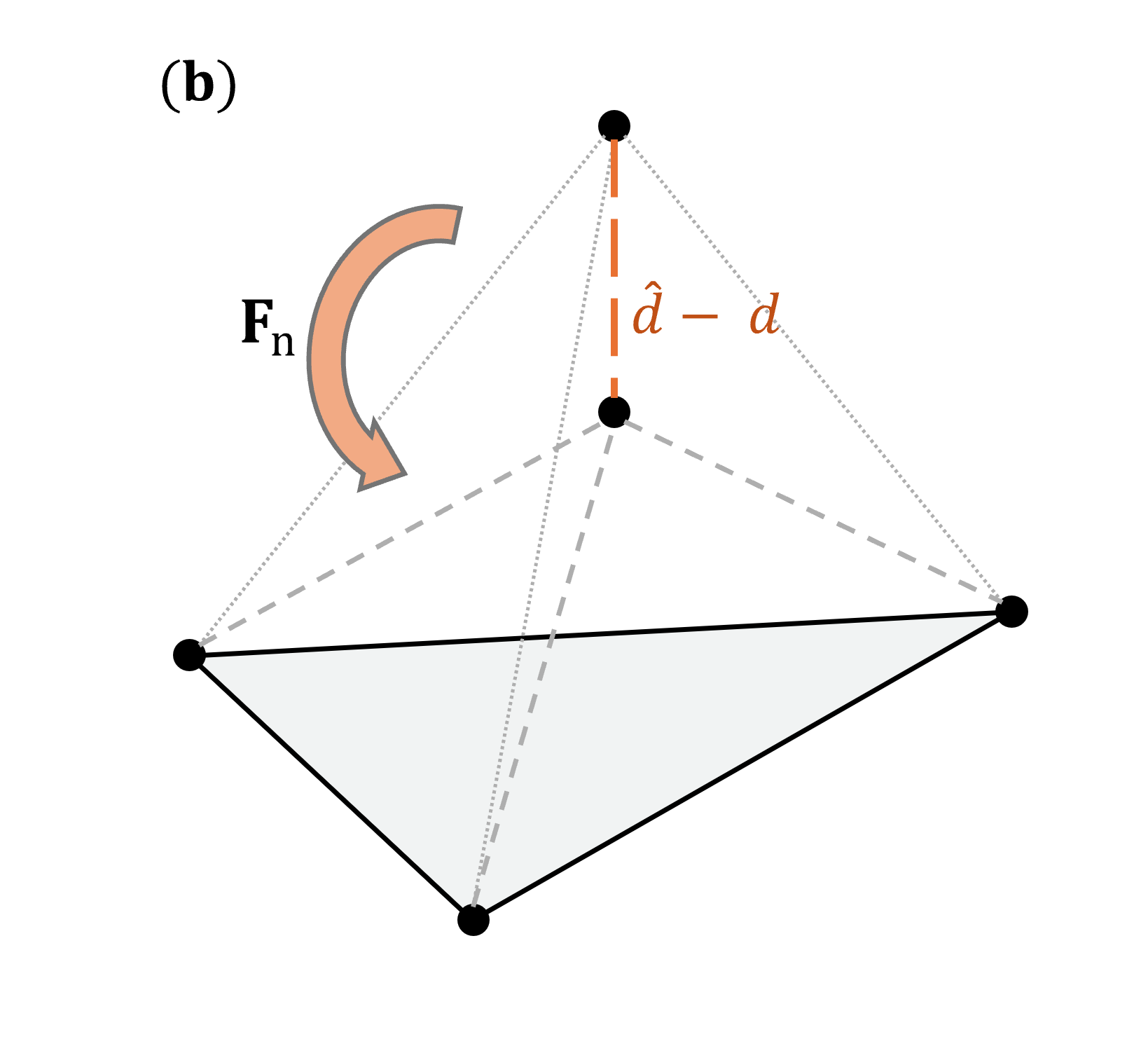}}
    {\includegraphics[width=0.32\linewidth]{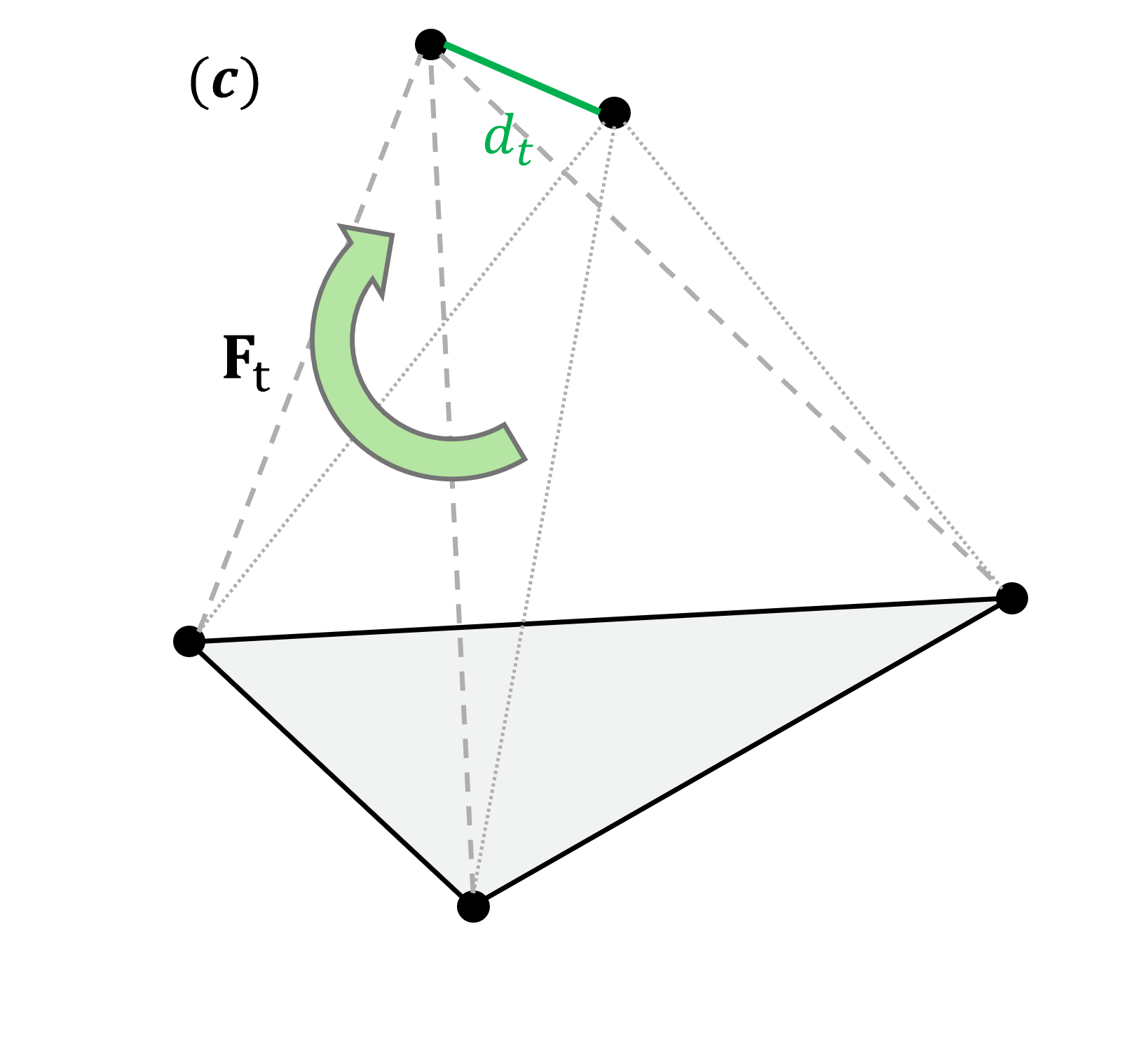}}
	\caption{Demonstration of virtual boundary elements for contact and friction. (a) the VBE in the reference configuration; (b) the VBE with normal deformation; (c) the VBE with shearing deformation.}
	\label{fig:vbe}
\end{figure}

To control the deformation along the normal and tangential directions separately, we use two separate energy density functions.
For the normal deformation, we apply the fourth invariant of the Cauchy-Green deformation tensor~\cite{Lin:2024:ARI} $I_4(\mathbf{n}) = ||\mathbf{F}\mathbf{n}||^2$, which is commonly used within fiber-reinforced materials~\cite{spencer1972deformations,wolper2020anisompm}, to define the energy density function:
\begin{equation}
{\Psi_{n}} = \frac{1}{2}\kappa C_n^2,{\kern 10pt}{C_n} = {\rm{sgn}}(J)\sqrt {||\mathbf{F}_n \mathbf{n}||^2}  - 1 \ge 0 ,
\end{equation}
where $\kappa$ denotes the stiffness, $\rm{sgn}(J)$ returns the sign of $J$.
For shearing deformation, we define the potential as a combination of the smoothed friction–velocity relation used in IPC~\cite{li2020incremental} and the fourth invariant of the corotated Cauchy–Green deformation tensor~\cite{Kugelstadt:2018:FCF}, defined as $I_4(\mathbf{n}) = ||(\mathbf{F} - \mathbf{R})\mathbf{n}||^2$.
For pure shearing deformation,  the potential can be further simplified by neglecting rotational effects and setting $\mathbf{R} = \mathbf{I}$, yielding the following expression:
\begin{equation}
\Psi_t =  \mu \lambda f(C_t), {\kern 10pt}C_t = {\sqrt {||(\mathbf{F}_t - \mathbf{I}) \mathbf{n}||^2} = 0}
\end{equation}
where $\lambda$ is the contact force magnitude, $\mu$ is the local friction coefficient, and \highlightv2{$f$ is a polynomial function defining a smooth displacement-based friction response.}
Taking derivative of both ${\Psi_{n}}$ and $\Psi_t$ with respect to $\mathbf{F}$ and multiplying both equations with $\mathbf{F}^{-1}$ yields the following first Piola-Kirchhoff stress tensor
\begin{equation}
(\mathbf{P}\mathbf{F}^{-1})_n = \kappa(1-\frac{\hat{d}}{d_n})\mathbf{n}\mathbf{n}^T, {\kern 5pt} (\mathbf{P}\mathbf{F}^{-1})_t = \mu \lambda f'(C_t)\mathbf{t}\mathbf{n}^T.
\end{equation}
\highlightv2{Finally, we apply analytical projection to both $(\mathbf{P}\mathbf{F}^{-1})_n$ and $(\mathbf{P}\mathbf{F}^{-1})_t$, whose formulations are}
\begin{equation}
\begin{array}{l}
\begin{aligned}
({\mathbf{P}}{{\mathbf{F}}^{ - 1}})_n^ +  &= \kappa {\mathbf{n}}{{\mathbf{n}}^T},{\kern 36pt}({\mathbf{P}}{{\mathbf{F}}^{ - 1}})_t^ +  = \mu \lambda {f}'({C_t})\frac{d_t}{\hat{d}}{\mathbf{I}},\\
({\mathbf{P}}{{\mathbf{F}}^{ - 1}})_n^ -  &= -\kappa \frac{\hat{d}}{d_n}{\mathbf{n}}{{\mathbf{n}}^T},{\kern 20pt}({\mathbf{P}}{{\mathbf{F}}^{ - 1}})_t^ -  = \mu \lambda {f}'({C_t})\left( {{\mathbf{t}}{{\mathbf{n}}^T} - \frac{d_t}{\hat{d}} \mathbf{I}} \right).
\end{aligned}
\end{array}
\end{equation}

For edge-edge (EE) collisions, the pairwise force is derived by replacing $\mathbf{n}$ with the normalized cross product of the two edge directions. 
A special case arises when the two edges are parallel; in this situation, the EE pair is discarded, as the resulting tetrahedron has zero volume and therefore does not contribute to boundary handling.
Compared with traditional collision-handling methods based on volume constraints~\cite{Muller:2015:AMR}, the anisotropic material model formulated using the fourth invariant ensures that collision responses and friction are strictly confined to the normal and tangential direction strictly, and do not introduce spurious tangential expansion in the material behavior.
\highlightv2{The scope and limitations of this smooth contact and friction formulation are discussed in Supplementary Material, Section~C.}

\begin{algorithm}[t]
\color{black}
\caption{Semi-Implicit Pairwise Descent (SIPD)}
\label{alg:sipd}
\footnotesize
\KwIn{Mesh $(\mathbf X,\mathcal T)$, surface $\mathcal S$, state $(\mathbf x^n,\mathbf v^n)$, time step $h$, iteration count $K$}
\KwOut{Updated state $(\mathbf x^{n+1},\mathbf v^{n+1})$}
$\mathbf y^* \leftarrow \mathrm{InertialPrediction}(\mathbf x^n,\mathbf v^n,h)$; $\mathbf y^0\leftarrow\mathbf y^*$\;
$\mathcal C\leftarrow\mathrm{DCD\text{-}BroadPhase}(\mathcal S,\mathbf x^n,\mathbf y^*)$\;
\For{$k\leftarrow0$ \KwTo $K-1$}{
  \If{$k\bmod4=0$}{
    $\mathcal A\leftarrow\mathrm{DCD\text{-}NarrowPhase}(\mathcal C,\mathbf y^k)$\;
    $\mathcal V\leftarrow\mathrm{BuildVBEs}(\mathcal A)$\;
  }
  $\mathbf g_i\leftarrow\frac{m_i}{h^2}(\mathbf y_i^k-\mathbf y_i^*)$; $\mathbf A_i\leftarrow\frac{m_i}{h^2}\mathbf I$\;
  \ForEach{$e\in\mathcal T\cup\mathcal V$ \textbf{in parallel}}{
    Compute $\mathbf F_e$, $\mathbf P_e\mathbf F_e^{-1}$, and its analytical split\;
    \ForEach{vertex pair $(i,j)$ of $e$}{
      Assemble $\mathbf L_{ij,e}$ into $\mathbf g_i$ and $\mathbf L^+_{ij,e}$ into $\mathbf A_i$\;
    }
  }
  $\Delta\mathbf y_i\leftarrow-\mathbf A_i^{-1}\mathbf g_i$ \tcp*[r]{independent $3\times3$ solves}
  $\alpha\leftarrow\mathrm{LineSearch}(\mathcal L,\mathbf y^k,\Delta\mathbf y)$; $\mathbf y^{k+1}\leftarrow\mathbf y^k+\alpha\Delta\mathbf y$\;
}
$\mathbf x^{n+1}\leftarrow\mathbf y^K$; $\mathbf v^{n+1}\leftarrow(\mathbf x^{n+1}-\mathbf x^n)/h$\;
\end{algorithm}
\raggedbottom
\section{Results and Discussion}

\begin{table*}[h]
\centering
\caption{Simulation parameters and performance metrics for various test scenarios.}
\label{tab:performance_results}
\resizebox{\textwidth}{!}
{
    \begin{tabular}{l rr ccc cc cc l}
    \toprule
    \multirow{2}{*}{{Example Name}} & \multicolumn{2}{c}{Number of} & \multicolumn{3}{c}{Material} & \multicolumn{2}{c}{Contact \& Friction} & \multicolumn{2}{c}{Solver Parameters} & \multicolumn{1}{c}{Time per step} \\
    \cmidrule(lr){2-3} \cmidrule(lr){4-6} \cmidrule(lr){7-8} \cmidrule(lr){9-10} \cmidrule(lr){11-11}
    & Vert. & Tet./Tri. & Type & $\rho$ & Stiffness & $\kappa_c$ & $\mu_c, \epsilon_v$ & $h$ & Iterations & avg/max \\
    \midrule
    \hyperref[fig:twist]{Twisting Beam} & 97,290 & 416,852 & Stable NeoHookean & 1000 & $E=1e8, \nu=0.48$ & $5 \cdot 10^3$ & $0.1; 10^{-2}$ & 3ms & 5 & 0.76ms/0.90ms \\
    \hyperref[fig:fall]{Dropping Bunny} & 225,316 & 1,034,600 & Stable NeoHookean & 1000 & $E=6e7, \nu=0.46$ & $5 \cdot 10^3$ & $0.1; 10^{-2}$ & 3ms & 4 & 2.8ms/3.1ms \\
    \hyperref[fig:cloth]{Cloth} & 9,100 & 17,590 & StVK & 1000 & $E=1e8, \nu=0.45$ & N/A & N/A & 1ms & 20 & 1.83ms/2.1ms \\
    \hyperref[fig:beamfall]{Beam Fall} & {463} & {1,479} & Stable NeoHookean & {1000} & $E \in \{4e5, 8e5,1.6e6, 8e6\}$ & {N/A} & {N/A} & {1ms} & {10/25/50} & avg: 0.3ms/0.7ms/1.3ms \\
    & & & & & $\nu=0.45$ & & & &  & max:0.3ms/0.8ms/2.1ms \\
    \hyperref[fig:friction]{Sliding Cube} & 24 & 30 & Stable NeoHookean & 1000 & $E=1e7, \nu=0.45$ & $5 \cdot 10^3$ & $\{0.0, 0.2, 0.4, 0.6, 0.8\}; 10^{-2}$ & 3ms & 3 & 0.02ms/0.03ms \\
    \hyperref[fig:noodle]{Noodle} & 697,005 & 1,849,230 & Stable NeoHookean & 100 & $E=1e8, \nu=0.46$ & $5 \cdot 10^3$ & $0.1; 10^{-2}$ & 3ms & 4 & 6.5ms/7.9ms \\
    \hyperref[fig:push]{Pushing bunny} & 32,188 & 147,800 & Stable NeoHookean & 1000 & $E=1e8, \nu=0.46$ & $10^4$ & $0.1; 10^{-2}$ & 3ms & 7 & 0.78ms/0.91ms \\
    \hyperref[fig:squishy]{Dropping Squishy Ball} & 167,445 & 478,500 & Stable NeoHookean & 1000 & $E=1e8, \nu=0.46$ & $10^4$ & $0.1; 10^{-2}$ & 3ms & 4 & 0.95ms/1.13ms \\
    \hyperref[fig:stretch]{Stretching Cube} & 64,000 & 355,914 & Stable NeoHookean & 1000 & $E=1e8, \nu=\{0.1,0.2,0.46,0.48,0.49,0.499\},$ & N/A & N/A & 3ms & 40 & 3.2ms/3.4ms \\
    \hyperref[fig:shape]{Bunny recover} & 3,219 & 14,308 & Corotated linear & 1000 & $E=1e8, \nu=0.45$ & N/A & N/A & 30ms & 20 & 0.28ms/0.32ms \\  
    \hyperref[fig:shape]{Octopus recover} & 3,430 & 13,578 & Stable NeoHookean & 1000 & $E=1e8, \nu=0.46$ & N/A & N/A & 3ms & 6 & 0.17ms/0.19ms \\ 
    \hyperref[fig:shape]{Armadillo recover} & 15,228 & 62,770 & Stable NeoHookean & 1000 & $E=1e8, \nu = 0.45$ & N/A & N/A & 3ms & 20 & 1.52ms/1.68ms \\    
    \bottomrule
    \end{tabular}
}
\end{table*}

\begin{table}[t]
\centering
\caption{Comparison of descent directions and required derivative information.}
\label{tab:compact_comparison_colored}
\footnotesize 
\begin{tabular}{l l l}
\toprule
\textbf{Method} & \textbf{Descent Direction} $\Delta \mathbf{y}_i$ & \textbf{Derivative Information} \\
\midrule
GD & $-\nabla \mathcal{L}_i$ & First order \\
Newton & $-[(\mathbf{M}/ h^2 + \mathbf{H})^{-1}\nabla \mathcal{L}]_i$ & Full Hessian $\mathbf{H}$ \\
VBD & $-(\mathbf{M}_{ii}/h^2 + \mathbf{H}_{ii})^{-1}\nabla \mathcal{L}_i$ & Vertex Hessian block $\mathbf{H}_{ii}$ \\
Jacobi-Pre. & $-(\mathbf{M}_{ii}/h^2 + \mathrm{diag}\{\mathbf{H}_{ii}\})^{-1}\nabla \mathcal{L}_i$ & Diagonal Hessian \\
\textbf{Ours} & $-(\mathbf{M}_{ii}/h^2 
+ \mathbf{L}^+_{ii})^{-1}\nabla \mathcal{L}_i$ & First-order $\mathbf{P},\mathbf{F}$ \\
\bottomrule
\end{tabular}
\end{table}

\highlightv2{Table~\ref{tab:compact_comparison_colored} summarizes the derivative information required by each descent direction. Newton, VBD, and Jacobi preconditioning construct full, block, or diagonal Hessian information from the constitutive derivative $\partial\mathbf{P}/\partial\mathbf{F}$, respectively, whereas our method constructs $\mathbf{L}_{ii}^{+}$ directly from the first-order quantities $\mathbf{P}$ and $\mathbf{F}$. Their measured convergence with respect to iteration count and wall-clock time is reported in Figure~\ref{fig:convergence}.}
\newline
We implement our solver and other comparison algorithms in CUDA on an NVIDIA RTX 5090 GPU using single-precision arithmetic. 
We adopt the stable NeoHookean model (Equation~\ref{eq:Ldecompose}) for both volumetric solids and cloth. 
Due to our method's low register usage (under 56 registers per thread), we maintain high GPU occupancy and use a thread block size of 256 for all experiments to maximize throughput. 
We set $\hat{d}=0.03$ for all collision \highlightv3{scenes}.

\subsection{Collision Handling}
We employ Discrete Collision Detection for collisions and friction. We restrict detection to the surface triangle meshes of our tetrahedral models. Each time step begins with a broad-phase filter to identify collision candidates. \highlightv1{During solver iterations, narrow-phase detection extracts active collision pairs to compute contact normals and penetration depths every 4 iterations.} We handle both Vertex-Face (V-F) and Edge-Edge (E-E) collisions, filtering the degenerate coplanar E-E cases. 

Our formulation treats collision pairs as "virtual tetrahedra," allowing them to be stored in the same element list as elastic tetrahedra. 
This enables a single CUDA kernel to process elasticity, collision, and friction by only varying the $\mathbf{PF}^{-1}$ matrix. 
Figures~\ref{fig:fall}, Figures~\ref{fig:noodle} and Figures~\ref{fig:twist} demonstrate the robustness and efficiency of this unified treatment.

\subsection{Positive definiteness}
\label{sec:PD}
\highlightv1{To validate our analytical projection strategy for SIPD, we test scenarios with high Poisson's ratios, where non-positive definiteness issues are more pronounced~\cite{chen2024stabler}.
As illustrated in Figure~\ref{fig:stretch}, under conditions of a high Poisson's ratio ($\nu > 0.48$) and large deformations, the non-positive definiteness problem becomes severe.
Directly using a non-positive definite matrix for iterations leads to divergence, as seen in the VBD method without positive definiteness enforcement, which fails in $\nu > 0.48$.
In contrast, our method maintains stability throughout the simulation.}

\highlightv2{Moreover, we evaluate our analytical projection strategy under stretching, compression, and bending with $\nu=0.499$, using the eigenvalue clamping and absolute eigenvalue projection strategies as baselines (Figure~\ref{fig:APD}). Our analytical projection accepts the full step $\alpha=1$ in most iterations across all three deformation modes and therefore requires little line-search backtracking.}

\subsection{Stress Test}
\highlightv1{We test our method under three extreme deformation scenarios to show the stability of SIPD in Figure~\ref{fig:shape}. 
The first one is a bunny model with randomly initialized positions around the surface of a sphere.
The second one is a highly flattened armadillo. 
The third is \highlightv3{an} octopus model with its arms strongly stretched.}

\highlightv1{Figure~\ref{fig:twist} is a challenging frictional contact case where two beams are twisted together. 
\highlightv2{Figure~\ref{fig:noodle} is a large-scale collision scenario in which 81 long noodles are dropped into a bowl, producing over one million contact pairs, substantially more than the 697K vertex degrees of freedom.}
Our method passes these stress tests without any stability issues, demonstrating the robustness of our SIPD method under extreme conditions.}

\subsection{Performance}
\label{sec:pfance}
\highlightv1{In this section, we evaluate our method's advantage by analyzing register usage, which critically determines GPU parallelism.
High register usage severely limits the number of active thread blocks on a Streaming Multiprocessor (SM), reducing GPU occupancy and its ability to hide memory latency.
For instance, within the bottleneck kernel for computing the gradient and descent direction, even the efficient VBD method requires 102 registers per thread (with the dPdF computation alone taking at least 81) and an occupancy of 33.3\%.
Conversely, our approach demands only 56, giving an occupancy of 66.7\%.
This substantial reduction in register usage directly translates to higher instruction throughput and better hardware utilization.
As a result, Figure~\ref{fig:TimeComp} demonstrates that the elasticity solver times of Our-GS and Our-Jacobi are only 67\% and 60\% of those of VBD and BlockJacobi, respectively.}
\begin{figure}[t]
\centering
\includegraphics[width=\linewidth]{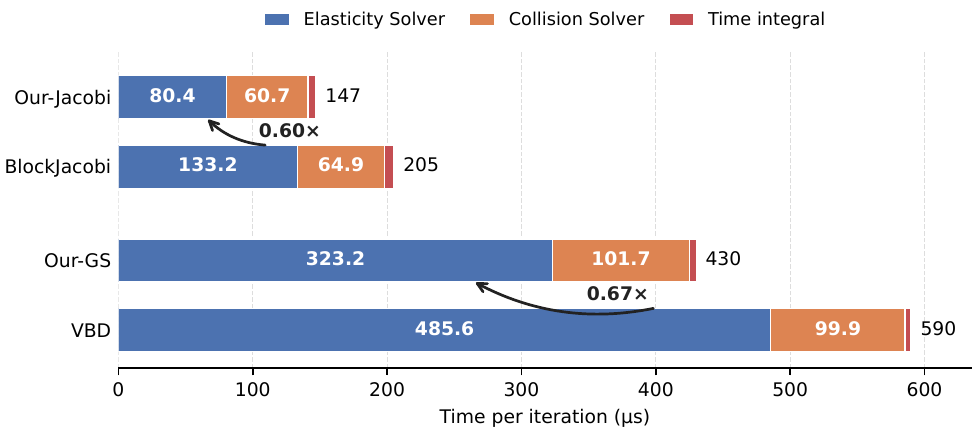}
\caption{Time overhead comparison in the Twisting Beam scene.}
\label{fig:TimeComp}
\end{figure}

\subsection{Comparisons with other descent methods}
By viewing the descent direction as $\Delta \mathbf{y}_i = -\mathbf{A} \nabla \mathcal{L}$, our solver effectively constructs a Hessian-free preconditioning matrix $\mathbf{A}$, and the traditional methods approximate the Hessian $\mathbf{H}$ directly.
While the Hessian matrix improves convergence in Newton-type methods, it introduces the heavy computational burden of fourth-order tensors $\partial \mathbf{P} / \partial \mathbf{F} \in \mathbb{R}^{9 \times 9}$ and positive definite projection steps. 
This limitation affects even efficient vertex-wise local descent approaches like VBD.
Our method, instead, derives the preconditioning matrix from second-order tensor $\mathbf{PF}^{-1} \in \mathbb{R}^{3 \times 3}$ without Hessian, and provides an SVD-free positive definite projection strategy. 
Although the omission of the original Hessian information degrades the per-iteration convergence rate, our method involves only second-order tensor computations and is completely SVD-free.
This significantly reduces the per-iteration overhead, achieving a highly effective trade-off that minimizes the total convergence time (\textbf{1.6x} faster than VBD).

\highlightv1{To prove this, we compare our method with several mainstream methods in table~\ref{tab:compact_comparison_colored}: 
Gradient Descent (GD) uses the identity matrix as the preconditioning matrix.
Newton-type Methods (e.g., full Newton~\cite{sin2013vega}, VBD, BlockJacobi~\cite{chen2024vertex}, Jacobi Preconditioning~\cite{wang2016descent}) involve constructing the global Hessian $\mathbf{H}$ or its block-diagonal counterparts $\mathbf{H}_{ii}$. 
To exclude the influence of graph coloring on convergence rates and execution time, we implement both our method and VBD in the Coloring Gauss-Seidel (GS) and Jacobi frameworks: Our-GS, Our-Jacobi, VBD and BlockJacobi.
Then, we compare the convergence behavior of these methods in Figure~\ref{fig:convergence} where a large compressed model is released.
As illustrated in Figure~\ref{fig:convergence}, although our methods (Our-GS and Our-Jacobi) exhibit a slower per-iteration convergence rate (left figure) compared to the VBD and BlockJacobi methods, they achieve a substantially shorter overall convergence time (right figure) due to a 33\%-40\% reduction in computational overhead brought by small register usage in Section~\ref{sec:pfance}.
}

\highlightv1{Furthermore, as noted by~\cite{chen2024vertex}, BlockJacobi (Jacobi framework) incurs significant line search overhead, resulting in slower overall performance compared to VBD (GS framework). 
Conversely, our method is highly compatible with the Jacobi framework. 
Our projection strategy requires markedly fewer line search iterations (see Figure~\ref{fig:APD}), enabling Our-Jacobi to outperform Our-GS.
}

\subsection{Comparisons with Peridynamics}
\highlightv1{
The mechanical response in traditional peridynamics relies heavily on the spatial distribution of the horizon, resulting in computational inaccuracies at physical boundaries caused by horizon truncation.
Our approach reformulates FEM using non-local concepts of peridynamics, achieving theoretical equivalence to standard FEM and fundamentally eliminating these boundary issues.}

\highlightv1{To verify this, we simulate a tensile test on a square plate with a central circular hole (dimensions: $10\text{cm}\times10\text{cm}$, discretization spacing $dx=0.185\text{cm}$, hole radius: $1\text{cm}$) in Figure~\ref{fig:hole}. 
The top and bottom two layers of particles were designated as traction boundaries, subject to an outward tensile velocity of $10\text{cm/s}$.
\highlightv2{We compare the velocity fields of three methods: traditional peridynamics, our SIPD, and standard FEM. For this setup, the peridynamics horizon size is set to $2\Delta x$ and the resulting system is solved using the method of~\cite{lu2023projective}; the FEM system is solved using a Newton method. Starting from the same initial state, traditional peridynamics and SIPD are each run for 100 iterations, while the Newton-based FEM solver is run for 20 iterations.}
The result shows that peridynamics exhibits non-physical discontinuities near the hole due to the aforementioned horizon truncation, whereas our method remains smooth.
Disregarding minor floating-point error, our method is computationally equivalent to the standard FEM.}

\section{Conclusion}
In summary, the proposed semi-implicit pairwise descent offers an efficient and robust solution to nonlinear optimization problems in physics-based simulation. 
By avoiding second-order derivatives and relying solely on first-order energy gradients, the method reduces computational complexity while remaining well suited for GPU parallelization. 
The reformulation within a nonlocal peridynamics framework enables a unified treatment of internal and boundary forces. The semi-implicit splitting and analytical projection strategy guarantees positive definiteness of the local stiffness matrices. 
As a result, the method supports stable, parallel particle updates with minimal overhead, making it particularly attractive for large-scale, real-time elastic simulations.

Despite the promising results, our approach has several limitations. 
Our contact and friction model remains dependent on mesh resolution, which may lead to spurious contact forces. 
\highlightv2{Although we revisit FEM from a nonlocal perspective, the current method does not yet support fracture simulation.}
\highlightv2{Our current penalty-based contact formulation does not strictly prevent penetration. To improve robustness in more challenging simulations involving extreme deformation and complex contact, a promising direction is to combine SIPD with IPC~\cite{li2020incremental}, using its barrier formulation and continuous collision detection to enforce nonpenetration while retaining SIPD's Hessian-free parallel updates. Developing a compatible formulation for IPC contact and friction without sacrificing GPU efficiency is left for future work.}

\begin{acks}
We thank the anonymous reviewers for their valuable comments and suggestions. This work was supported by the New Generation Artificial Intelligence--National Science and Technology Major Project of China (No.~2025ZD0123902), the National Natural Science Foundation of China (Nos.~92570206 and 62302490), and the Basic Research Project of ISCAS (No.~ISCAS-JCMS-202403).
\end{acks}
\bibliographystyle{ACM-Reference-Format}
\bibliography{sample-base}

\appendix

\clearpage
\begin{figure*}
	\centering
    \includegraphics[width=0.85\linewidth]{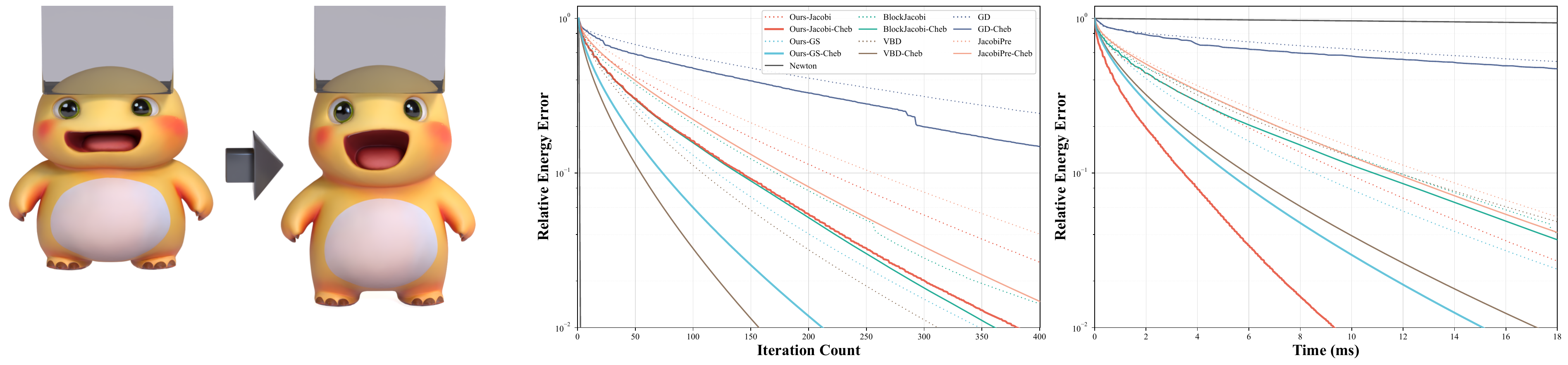}
	\caption{\textbf{Comparison of various descent methods}. We evaluate algorithm performance using the Baby Dragon model (9k vertices, 38k tetrahedra). The model's top vertices are fixed within a glass block, and an initial compression is applied by pushing the foot vertices upward (left). The positional constraints on the feet are then instantaneously removed to simulate deformation over a single time step of $h = 33$ ms (right). The plots illustrate the relative energy error as a function of iteration count and computation time. The relative energy error is calculated by$\frac{|\mathcal{L}^k - \mathcal{L}^*|}{|\mathcal{L}^0 - \mathcal{L}^*|}$, where $\mathcal{L}^*$ is calculated by the full Newton method. All methods are implemented in a unified GPU framework, with the Chebyshev acceleration parameter set to $\rho = 0.95$.}
	\label{fig:convergence}
\end{figure*}
\begin{figure*}
    \includegraphics[width=0.204\linewidth]{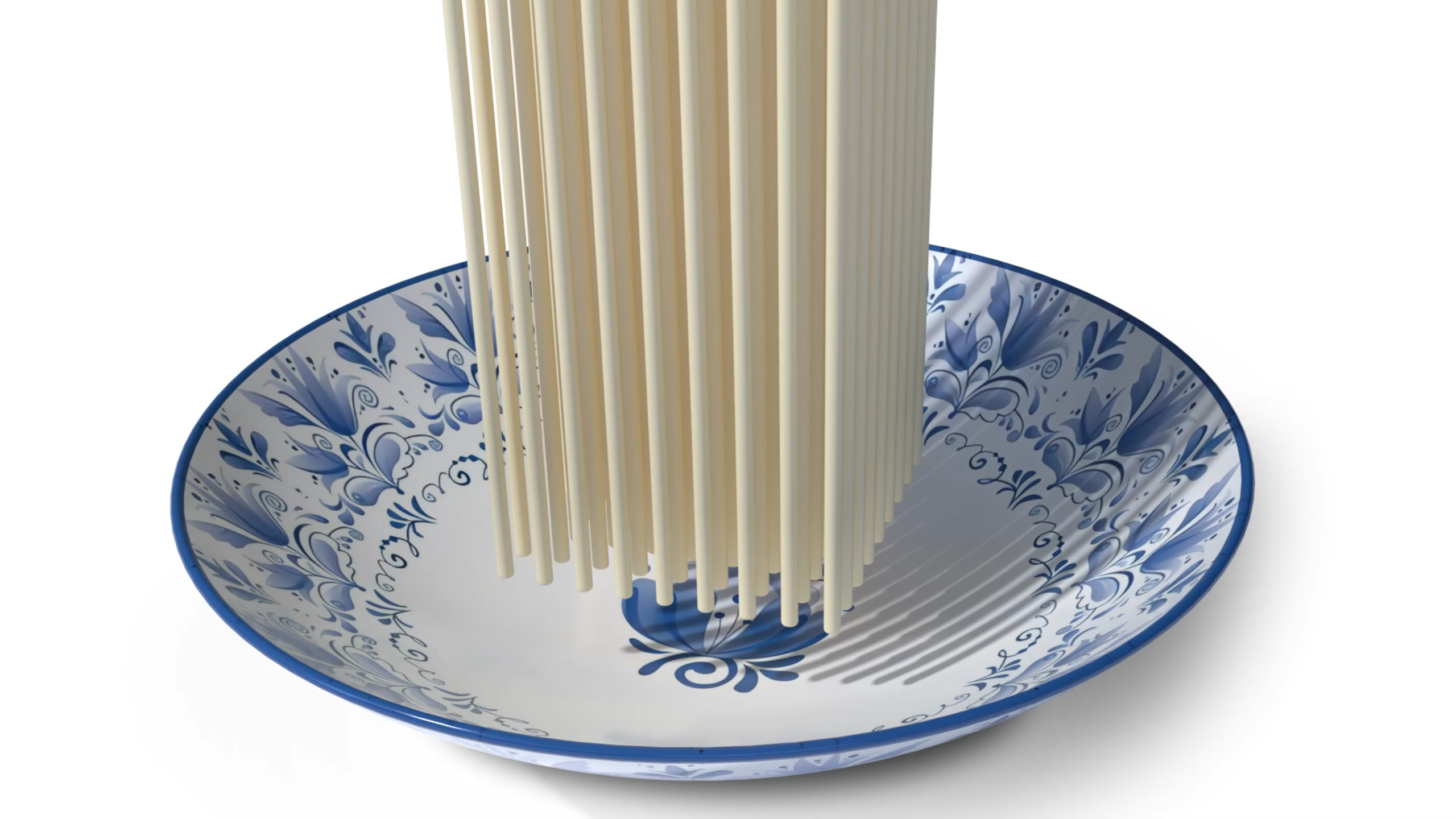}
    \includegraphics[width=0.204\linewidth]{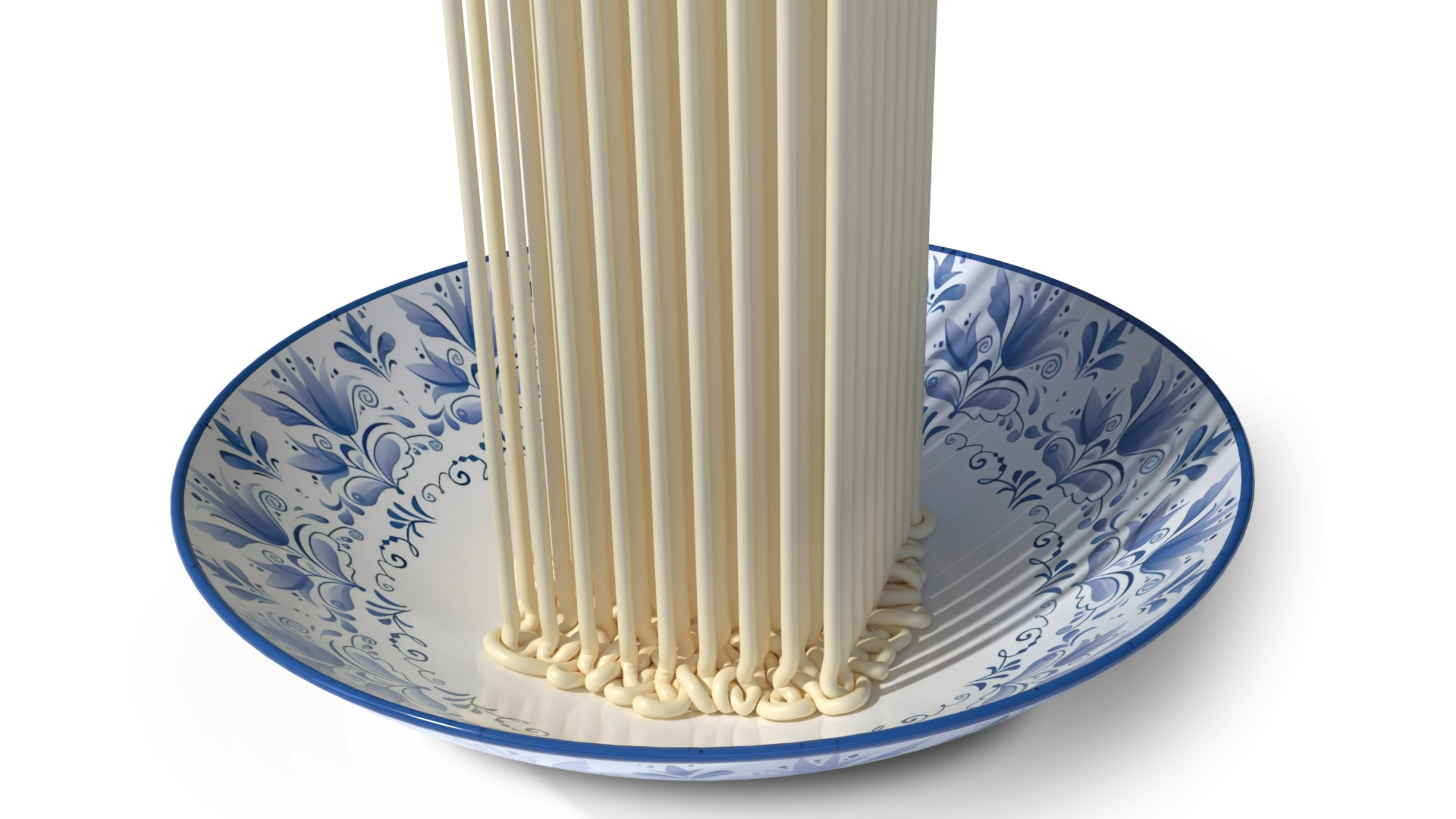}
    \includegraphics[width=0.204\linewidth]{images/noodle/noodle06.png}
    \includegraphics[width=0.204\linewidth]{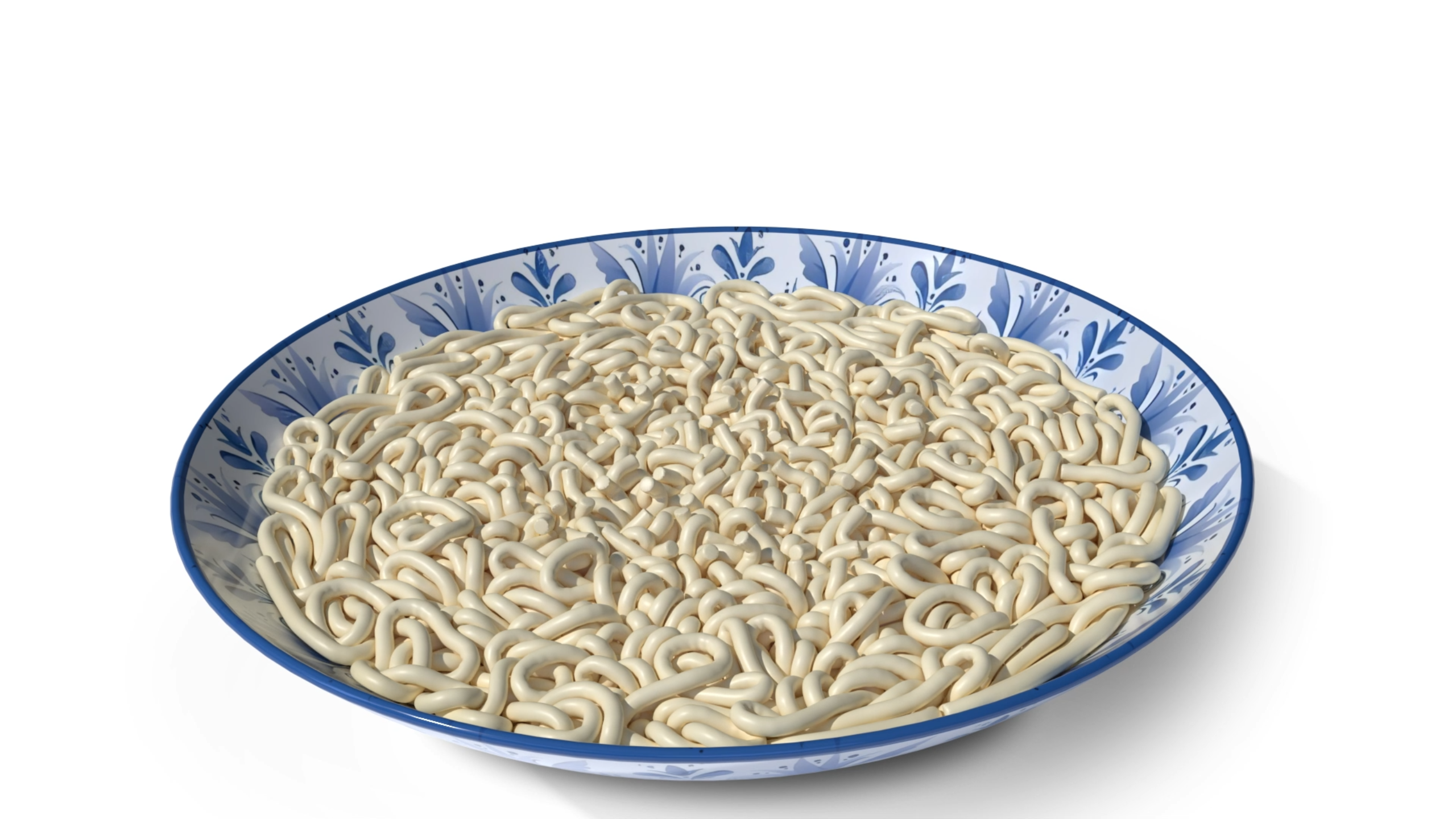}
	\caption{\textbf{Noodle Simulation}. 81 long noodles (1.8M tetrahedra) dropping into a porcelain bowl. Our animation is plausible and free of any collision artifact, with an average cost of 14.3ms per step ($h = 3ms$).}
	\label{fig:noodle}
\end{figure*}
\begin{figure*}
\includegraphics[width=0.68\linewidth]{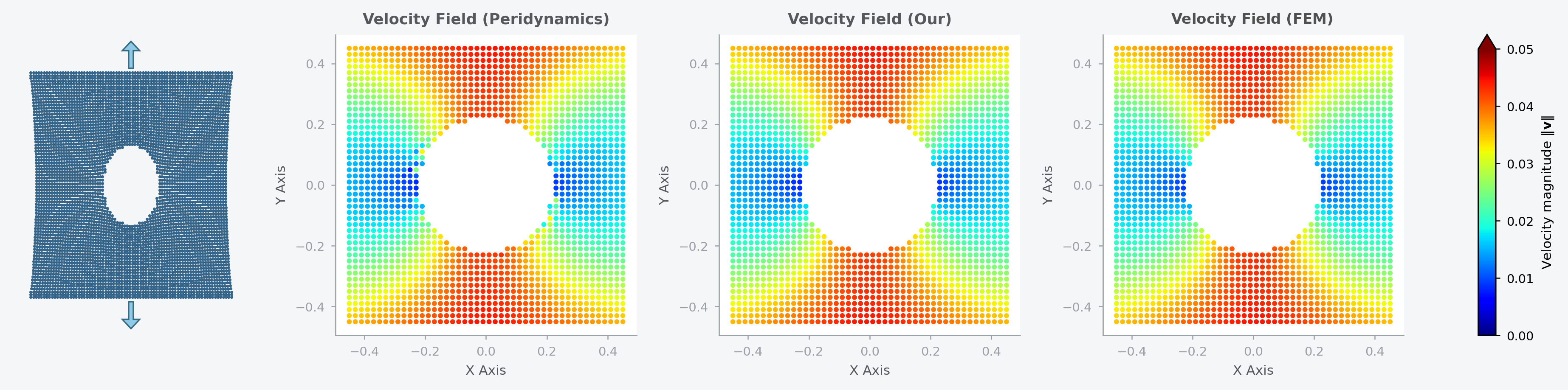}
\caption{\textbf{Open-hole Tension Test}. We compare the velocity \highlightv3{fields} of three \highlightv3{methods}: peridynamics, \highlightv3{our method}, and \highlightv3{standard} FEM. Peridynamics presents anomalous discontinuities around the hole boundary, an issue entirely absent in both our method and FEM.}
\label{fig:hole}
\end{figure*}
\begin{figure*}
	\includegraphics[width=0.638\linewidth]{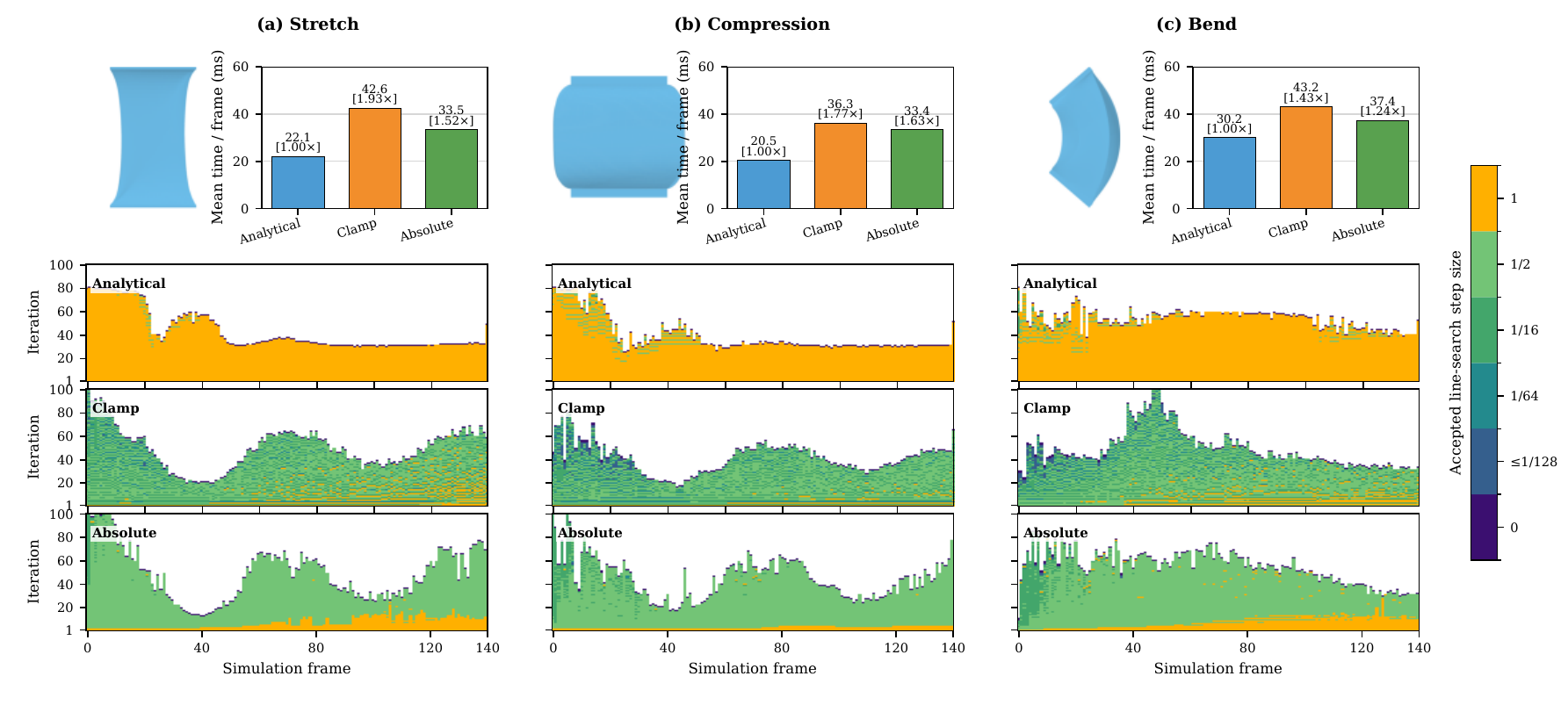}
	\caption{\highlightv2{\textbf{Comparison of three positive-definitization strategies over complete simulations} with $\nu=0.499$: (a) stretching, (b) compression, and (c) bending. The top row shows the deformation and mean valid time per frame, while the heatmaps below show the accepted step size over frame and iteration. All runs terminate at a relative energy error below $10^{-4}$ or after 100 iterations.}}
	\label{fig:APD}
\end{figure*}

\begin{figure}[t]
    \centering
    \includegraphics[width=0.82\linewidth]{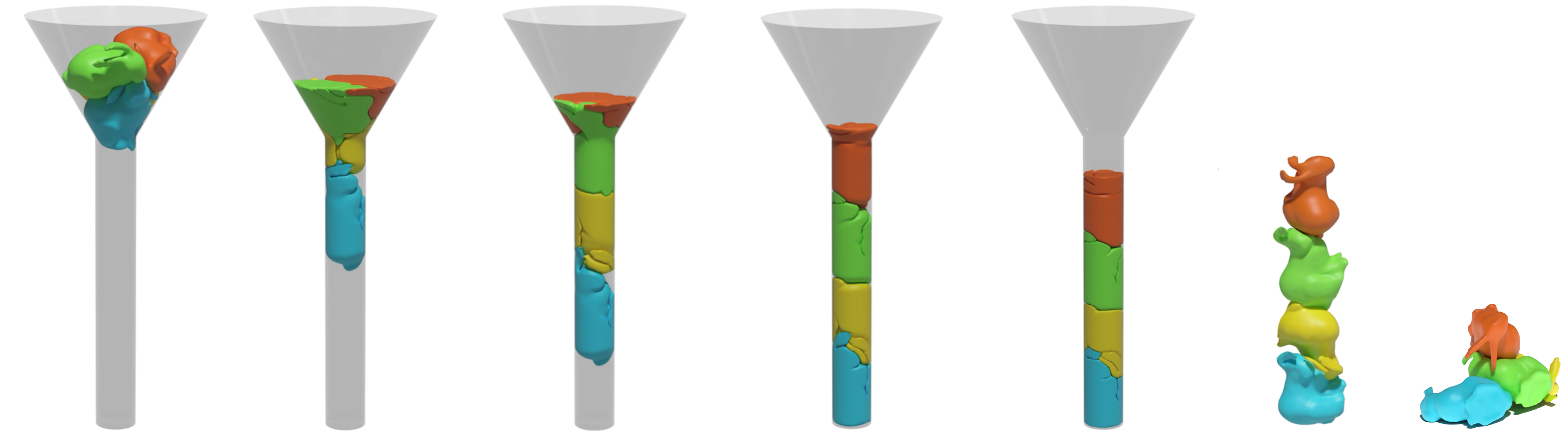}
    \caption{\textbf{Pushing Bunny}.4 bunnies are compressed into a narrow cylindrical cavity, followed by a sudden removal of the confinement.}
    \label{fig:push}
\end{figure}

\begin{figure}[t]
    \centering
    \includegraphics[width=0.32\linewidth]{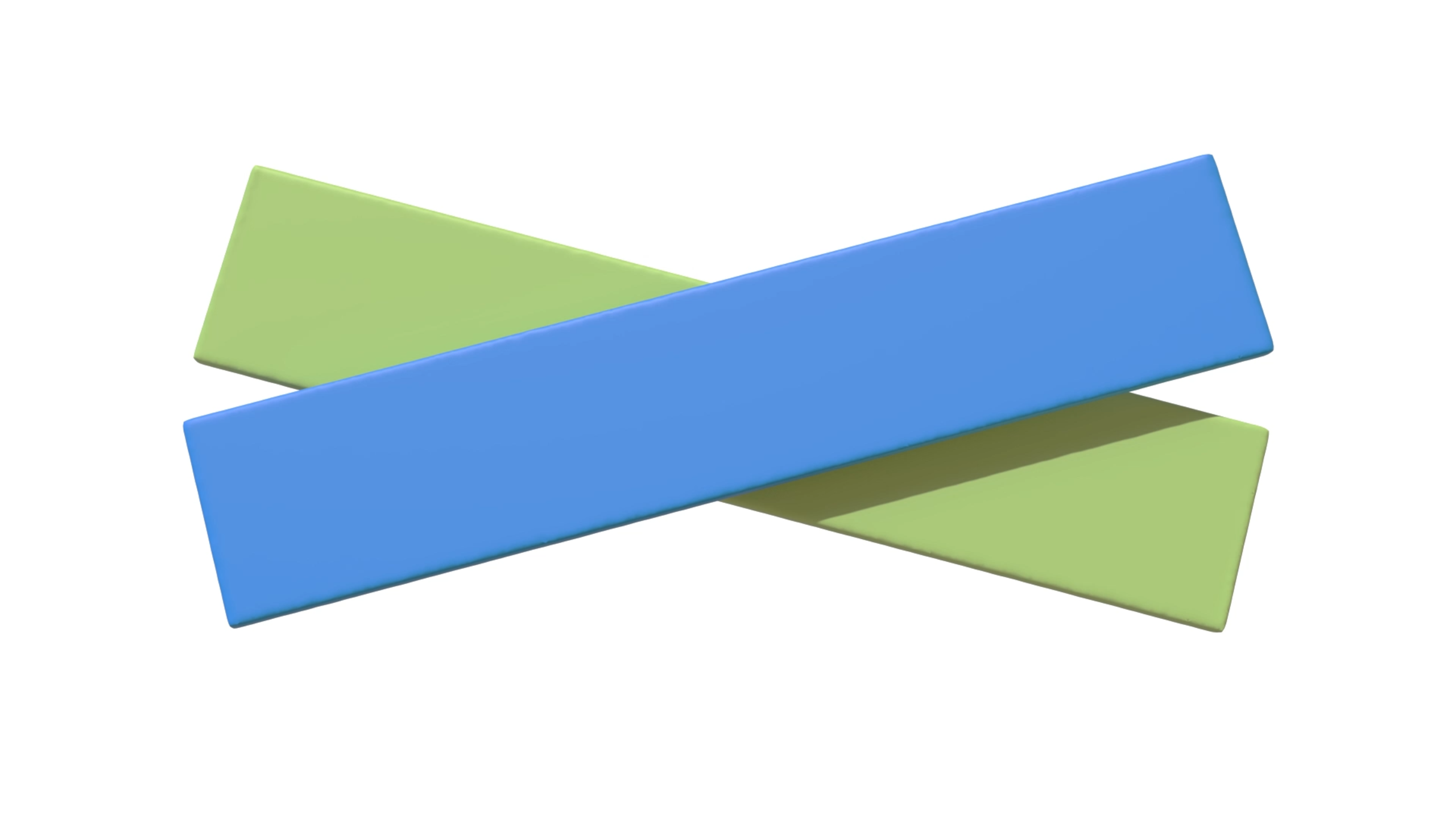}
    \includegraphics[width=0.32\linewidth]{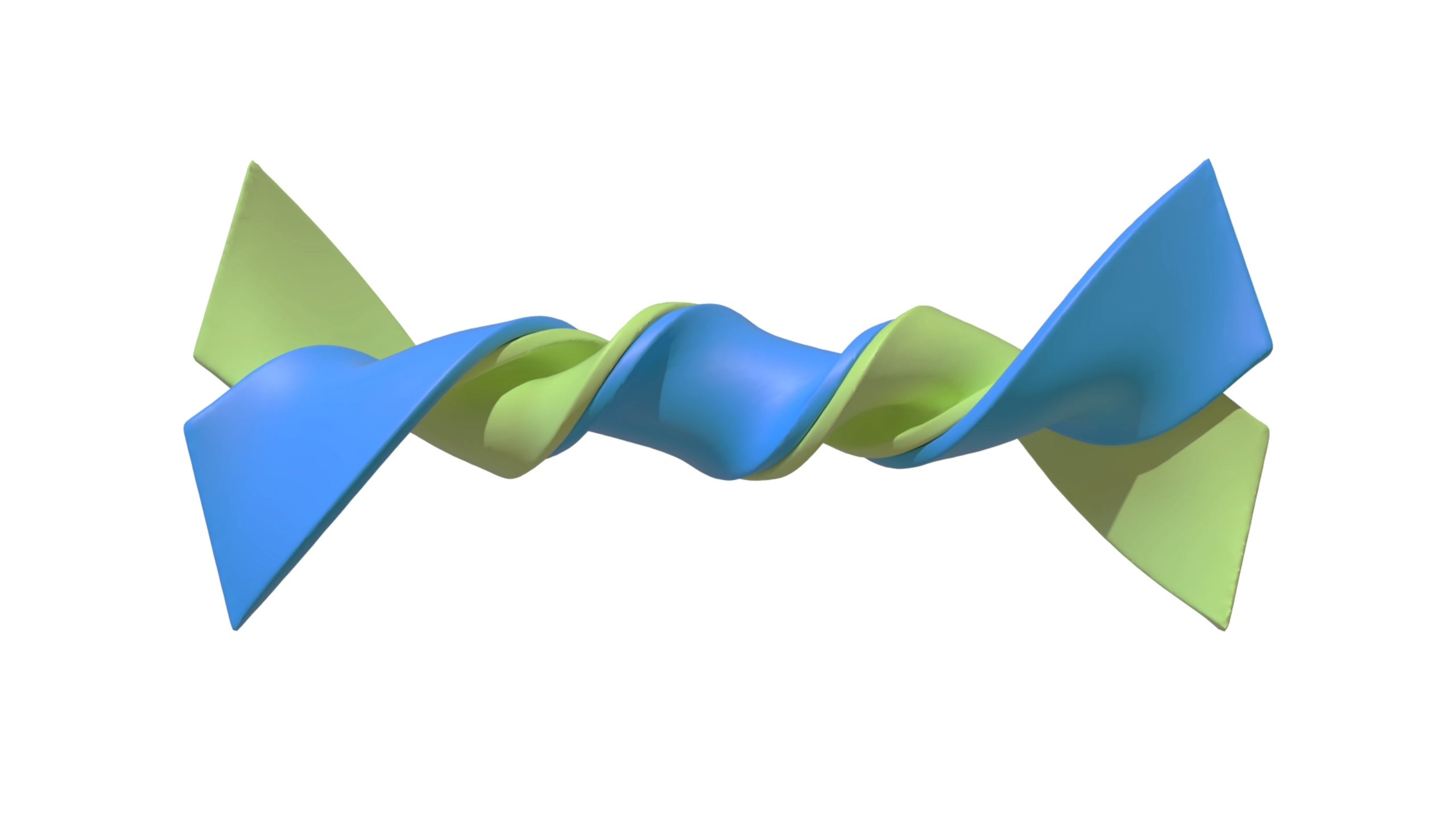}
    \includegraphics[width=0.32\linewidth]{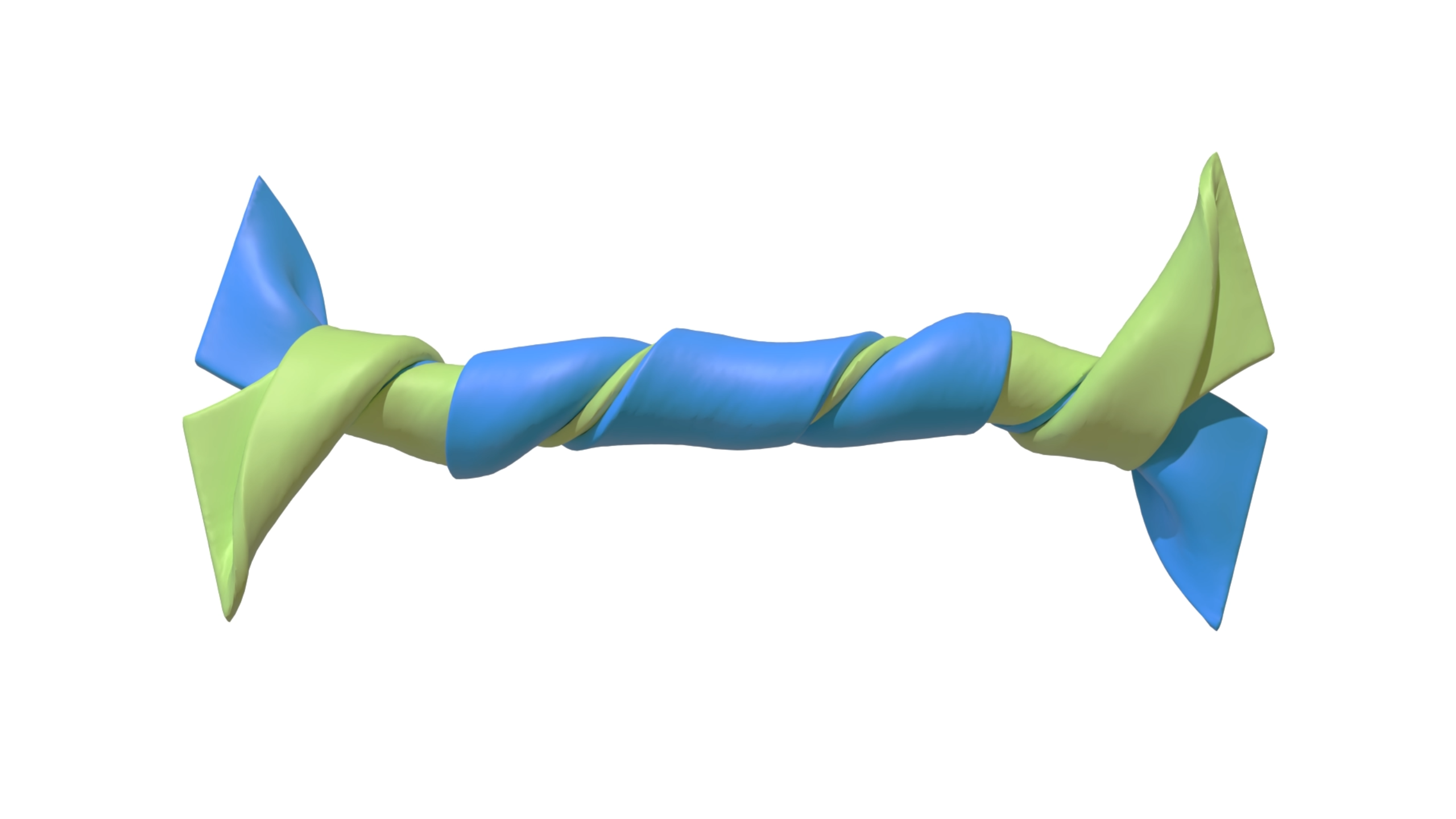}
    \caption{\textbf{Twisting Beam}. Two beams (97K vertex, 416K tetrahedra) intertwined under complex frictional contact and buckling.Our solver delivers artifact-free results at a near real-time rate, averaging 3.5 ms per step ( $h = 3 ms$ )}
    \label{fig:twist}
\end{figure}

\begin{figure}
	\centering
    \includegraphics[width=0.9\linewidth]{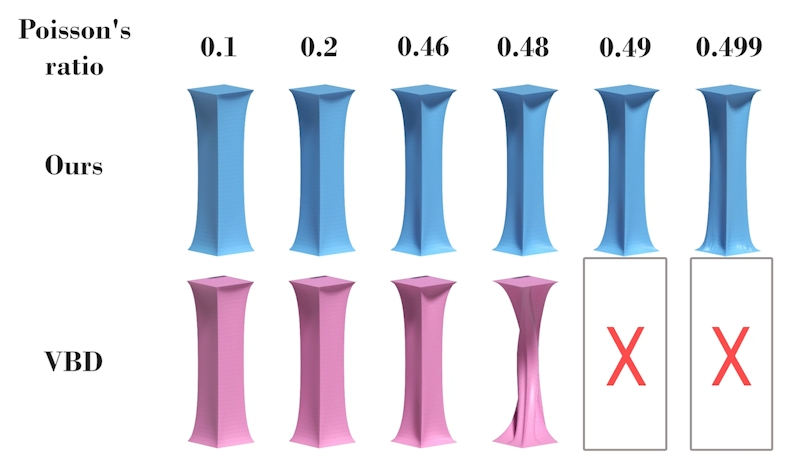}
	\caption{\textbf{Stretching Cube}. A cube (64K vertex, 355K tetrahedra) is stretched under a varying Poisson's ratio. Our solver delivers artifact-free results from $\nu = 0.1$ to $\nu = 0.499$, while the VBD solver fails when $\nu > 0.48$.}
	\label{fig:stretch}
\end{figure}

\begin{figure}[htbp]
	\centering
    \includegraphics[width=0.75\linewidth]{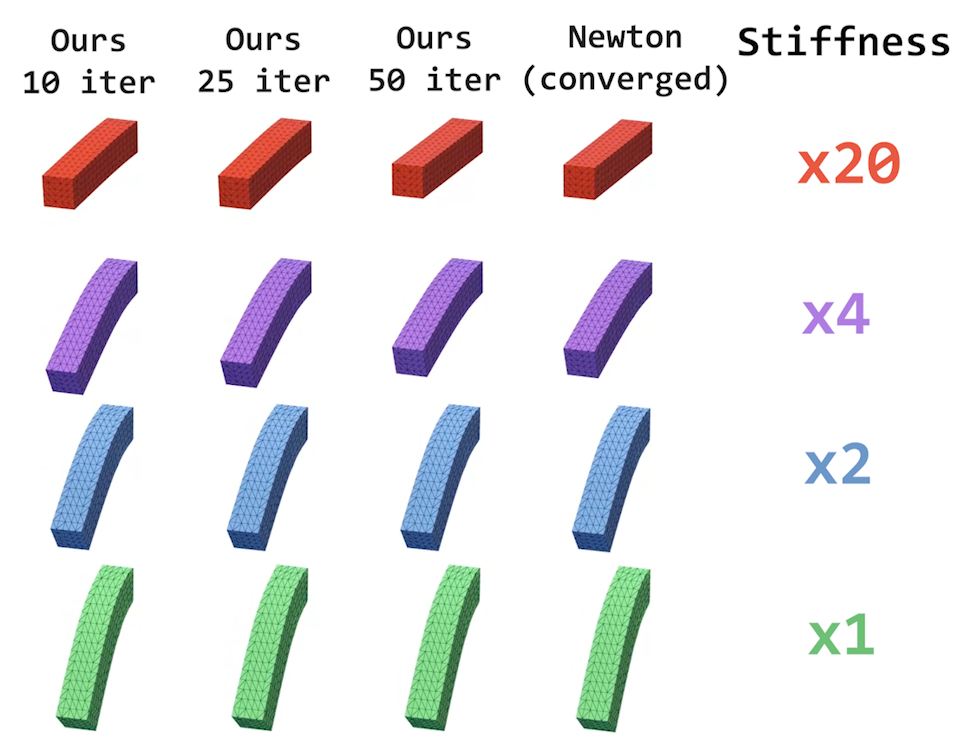}
	\caption{\textbf{Beam fall}. Visual convergence behavior of a beam (463 vertices, 1.5K tetrahedra), illustrated through simulations with varying material stiffness and iterations per frame. }
	\label{fig:beamfall}
\end{figure}

\begin{figure}[htbp]
	\centering
    \includegraphics[width=0.58\linewidth]{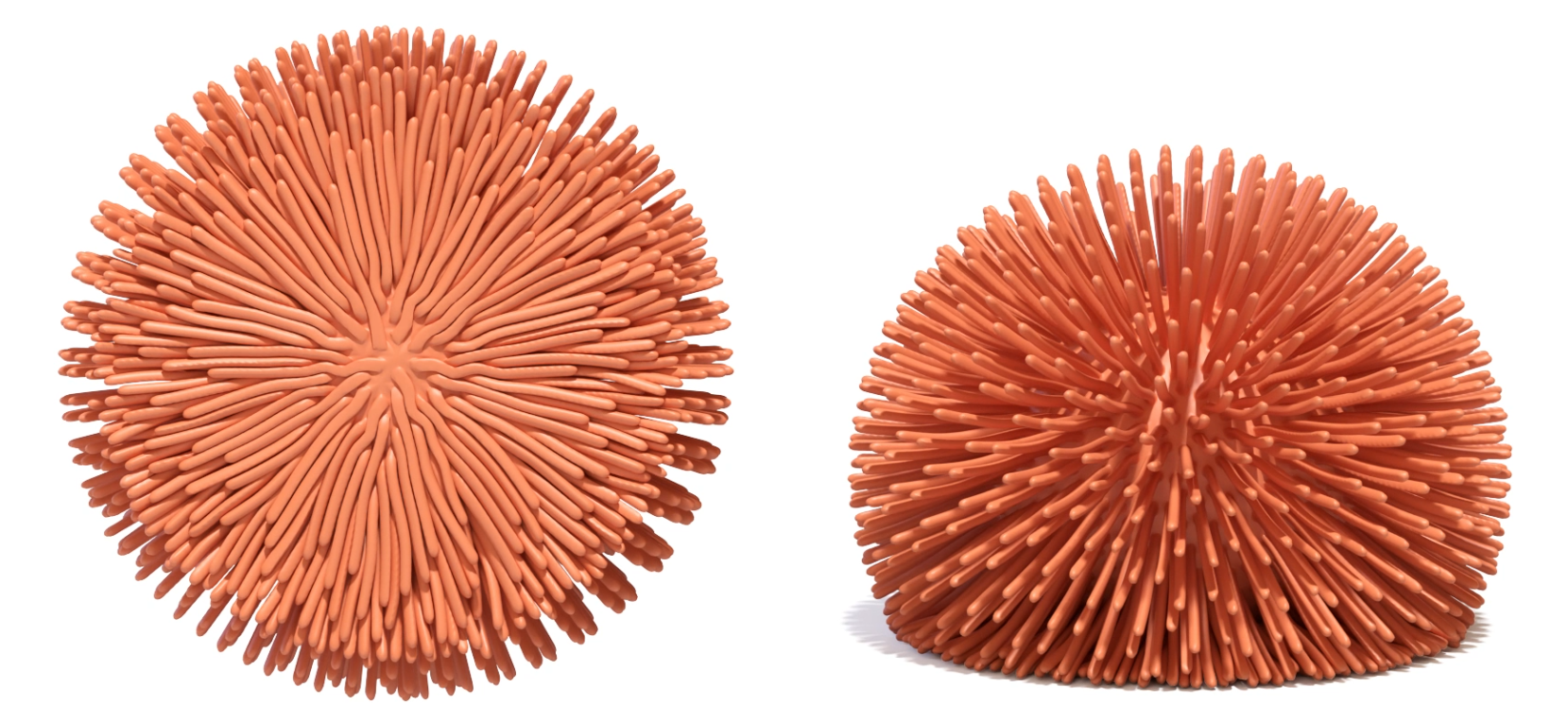}
	\caption{\textbf{Squishy Ball}. A squishy ball with tentacles (167.4K vertices, 478.5K tetrahedra) is simulated as it drops and impacts the ground, with a per-frame compute time of 4.3 ms and a timestep of 3 ms. }
	\label{fig:squishy}
\end{figure}

\begin{figure}[htbp]
	\centering
    \includegraphics[width=0.62\linewidth]{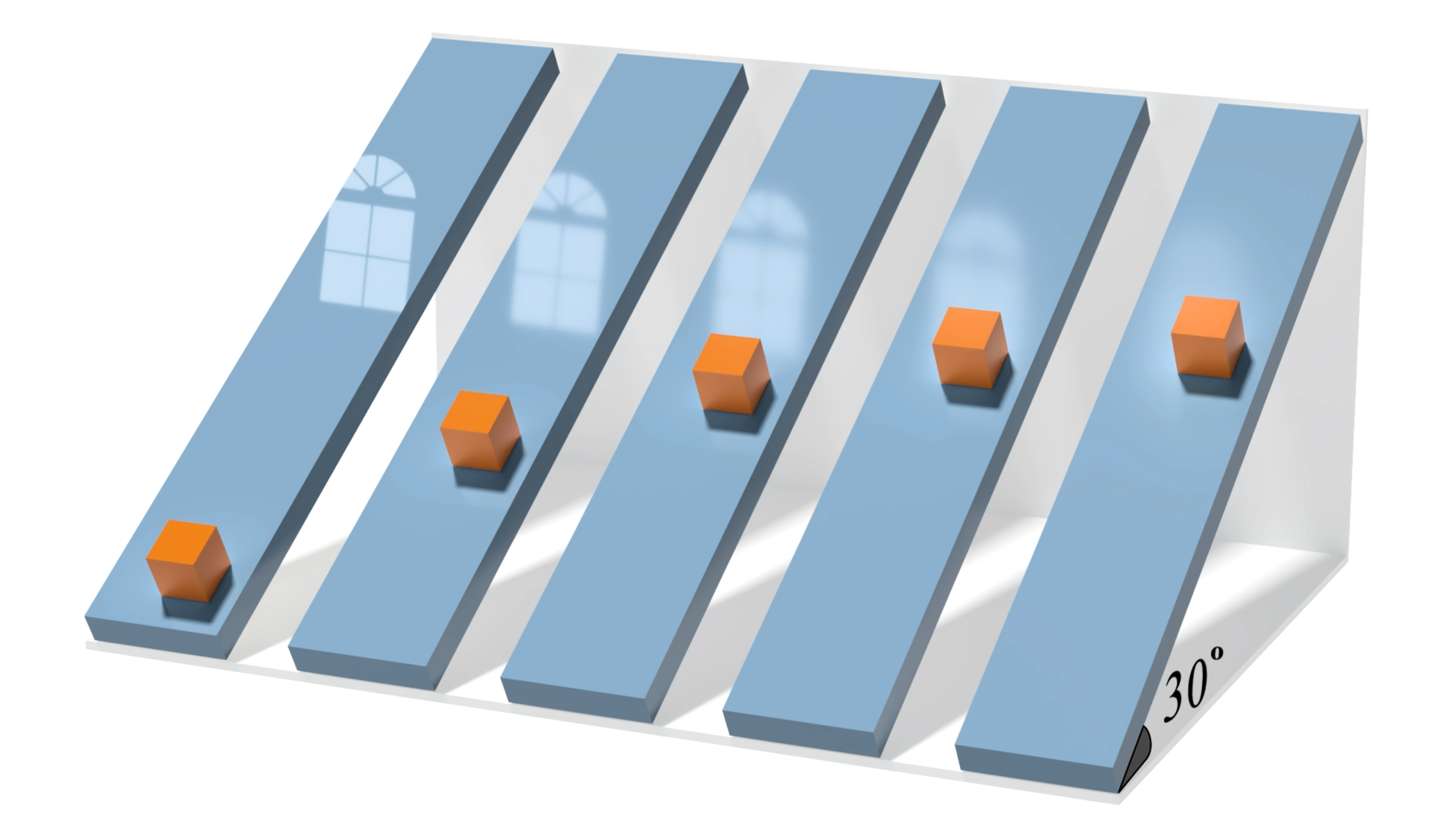}
	\caption{\textbf{Sliding Cube}. We test different friction coefficients $\mu_c$ for a cube sliding down an inclined plane. From left to right: $\mu_c = 0.0, 0.2, 0.4, 0.6, 0.8$. Our method can robustly handle various frictional contacts. }
	\label{fig:friction}
\end{figure}

\begin{figure}[htbp]
	\centering
    \includegraphics[width=0.26\linewidth]{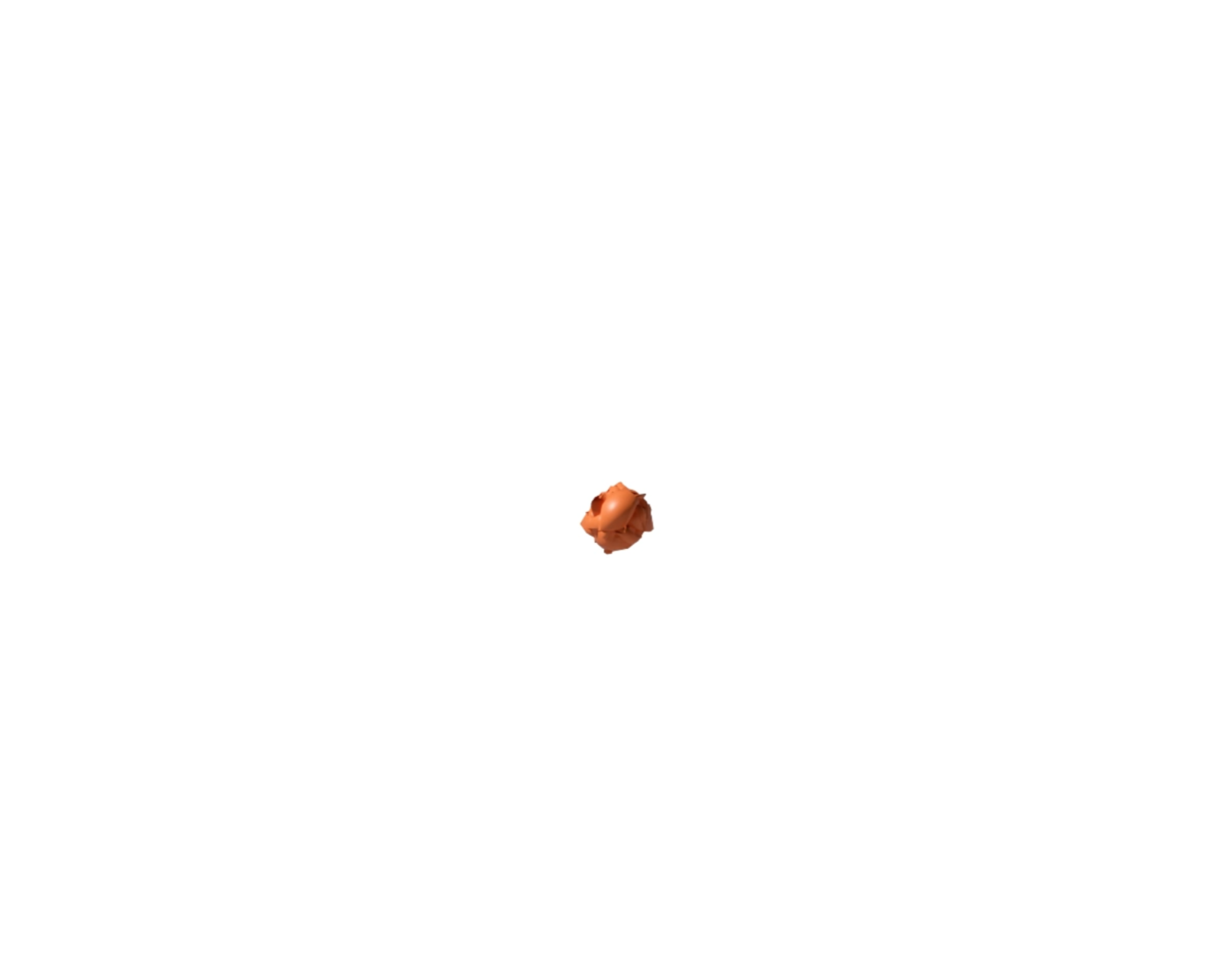}
    \includegraphics[width=0.26\linewidth]{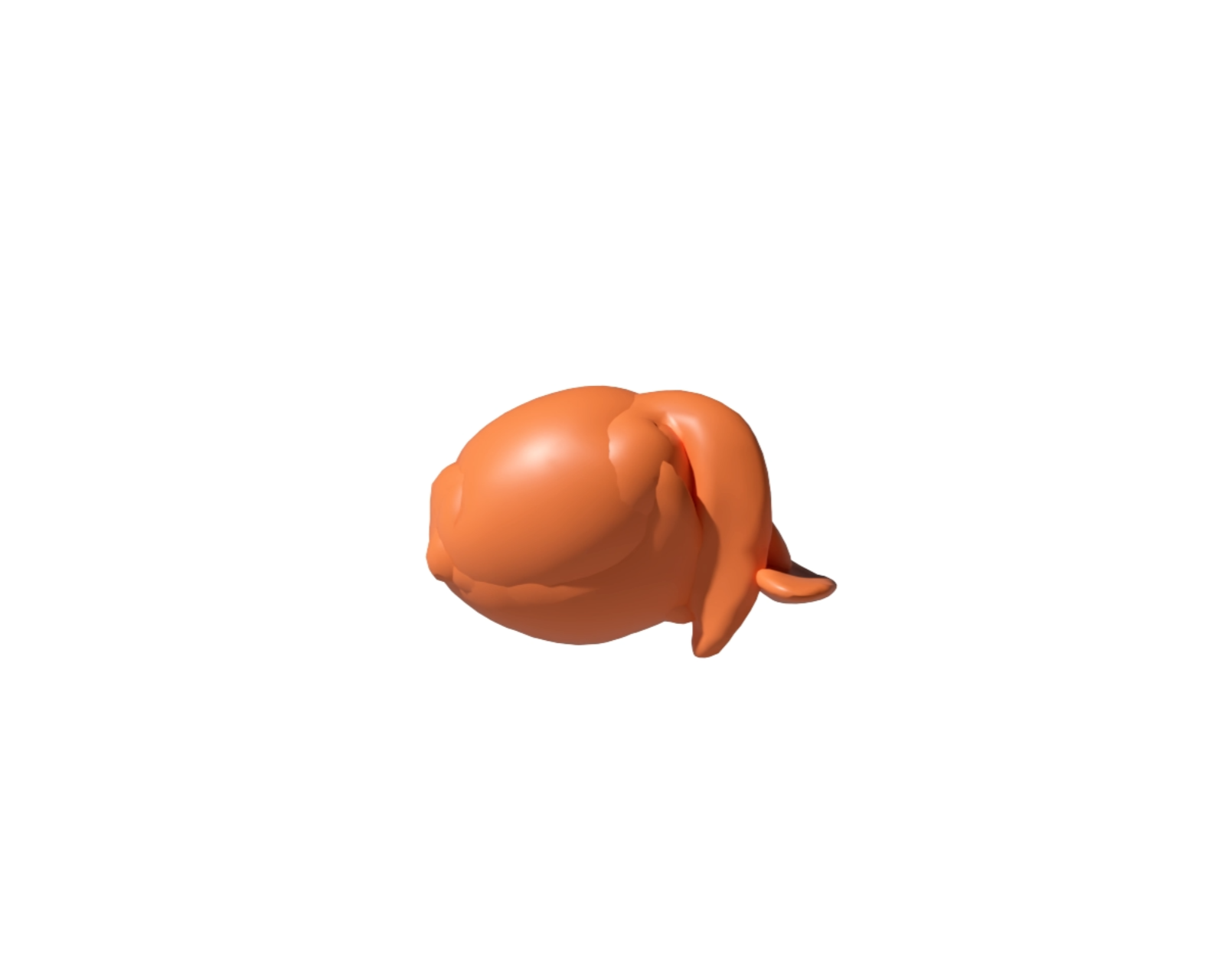}
    \includegraphics[width=0.26\linewidth]{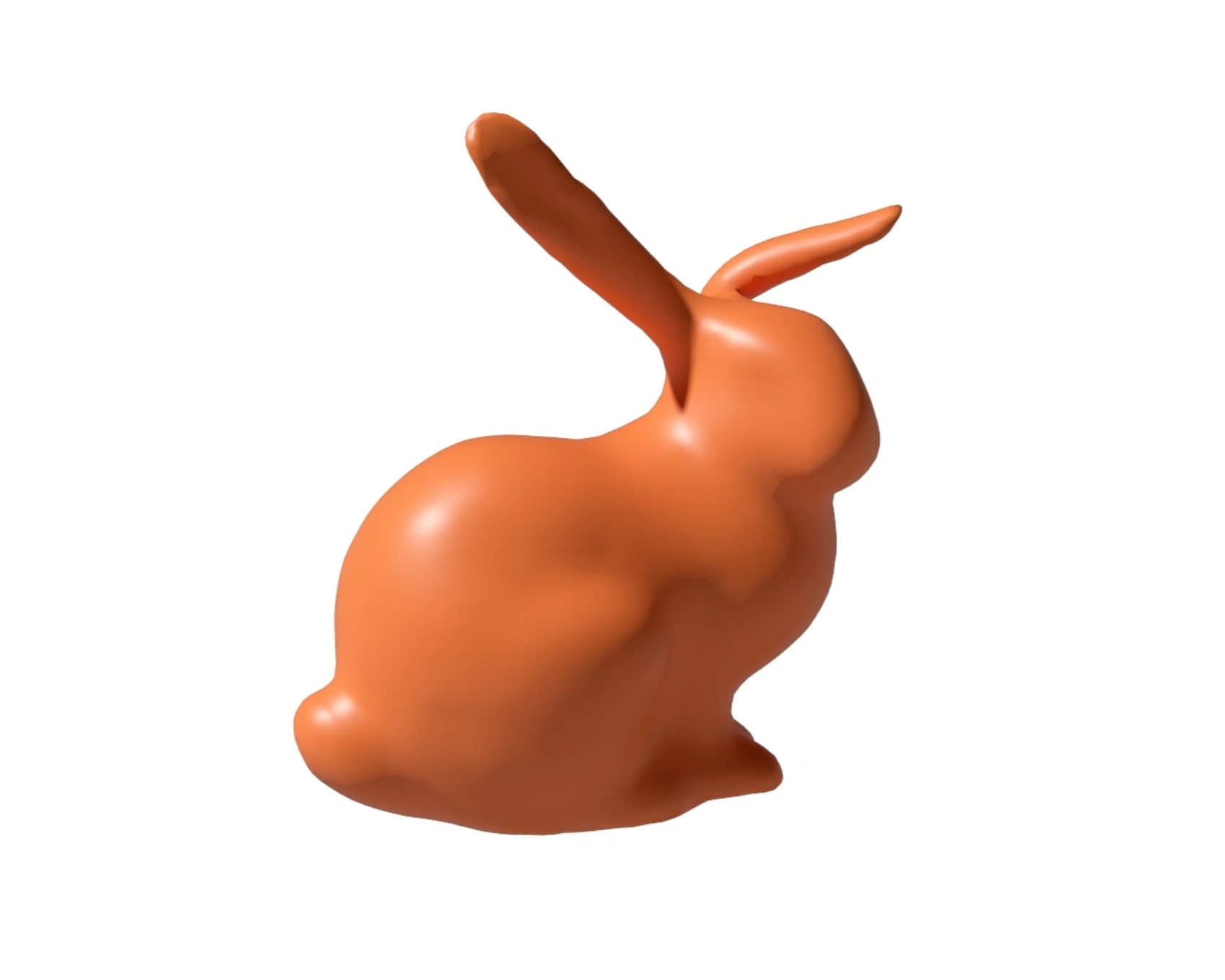}
    \includegraphics[width=0.76\linewidth]{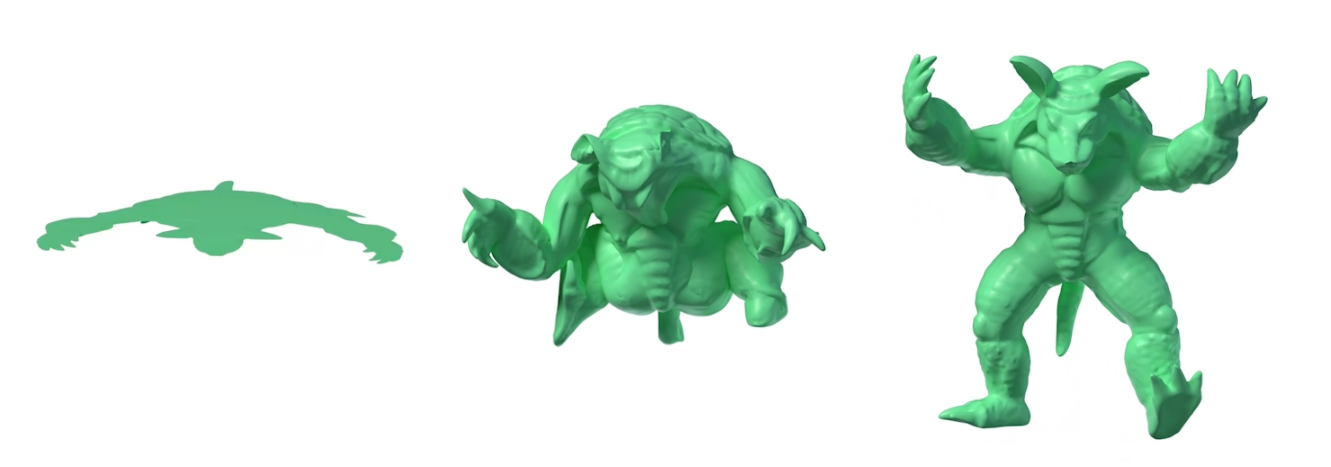}
    \includegraphics[width=0.76\linewidth]{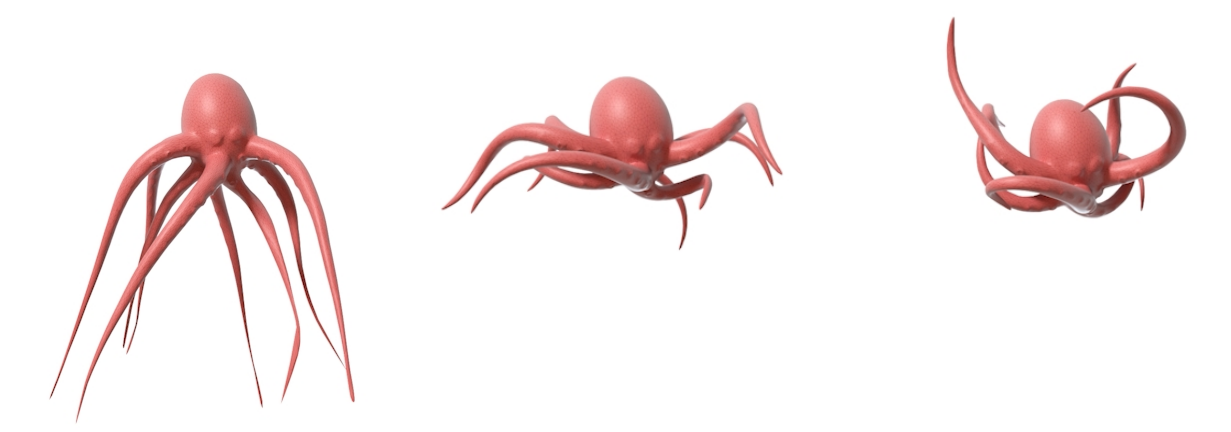}
	\caption{\highlightv2{\textbf{Stress Test}. Robustness to extreme initialization: \textbf{(Top)} a randomly initialized bunny model (3K vertices, 14K tetrahedra) returns to its reference configuration. \textbf{(Middle)} a highly flattened armadillo (15K vertices, 62K tetrahedra) recovers its original shape without artifacts. \textbf{(Bottom)} an octopus model (3K vertices, 13K tetrahedra) with eight tentacles recovers from a severely stretched state.}}
	\label{fig:shape}
\end{figure}

\begin{figure}[htbp]
	\centering
    \includegraphics[width=0.55\linewidth]{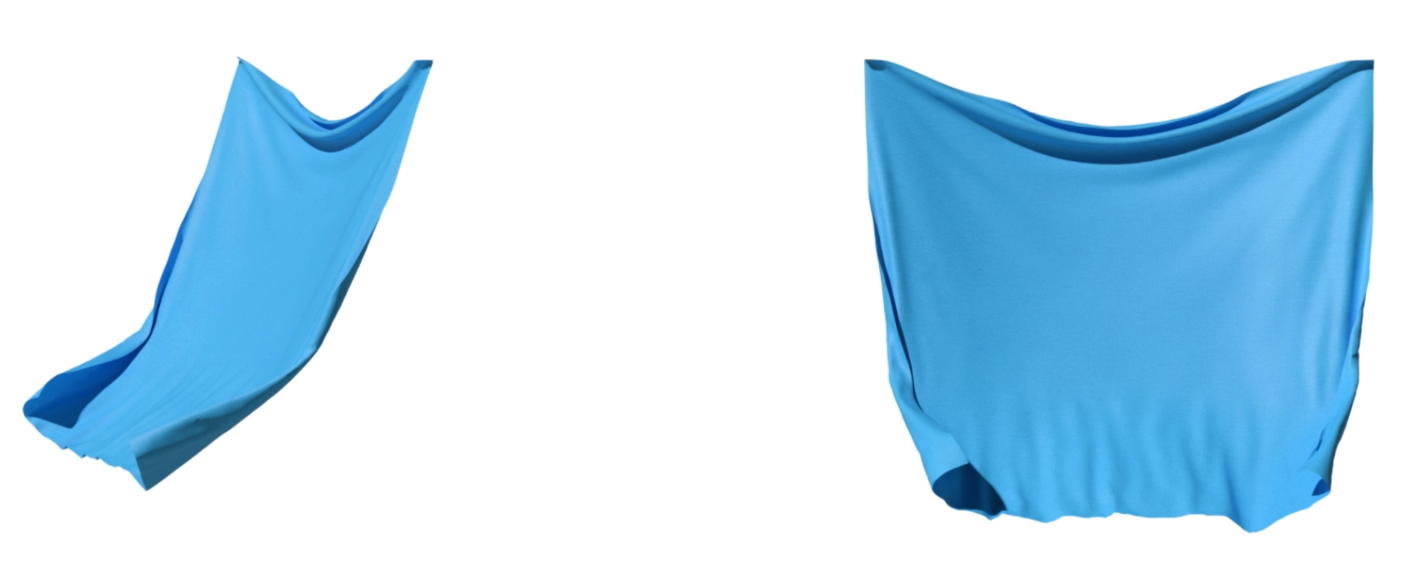}
	\caption{\textbf{Cloth}. Draping of a cloth with two fixed points. Our current implementation models cloth as a hyperelastic membrane (without bending). }
	\label{fig:cloth}
\end{figure}
\clearpage
\end{document}